\documentclass[12pt]{article}
\usepackage{a4wide}
\usepackage{latexsym}
\usepackage{amsmath}
\usepackage{amssymb}
\usepackage{amsfonts}
\usepackage{amscd}
\usepackage{amsthm}
\usepackage{cite}
\usepackage{placeins}
\usepackage{axodraw2}

\usepackage{pslatex}
\usepackage{graphicx}
\usepackage[latin1,utf8]{inputenc}
\usepackage[T1]{fontenc}

\usepackage{mathtools}
\usepackage{mathdots}
\usepackage{bm}
\usepackage{nomencl}

\allowdisplaybreaks

\newcommand{\bq}{\begin{eqnarray}}
\newcommand{\eq}{\end{eqnarray}}

\newcommand{\eps}{\varepsilon}

\newcommand{\kinvar}{y}

\newcommand{\laportaorder}{(a,w,o,|\mu|,\dots)}

\newcommand{\diff}{{d}}

\newcommand{\Divisor}{P}

\newcommand{\preabs}{C_{\mathrm{abs}}}
\newcommand{\prerel}{C_{\mathrm{rel}}}
\newcommand{\preclutch}{C_{\mathrm{clutch}}}
\newcommand{\prebaikov}{C_{\mathrm{Baikov}}}

\newcommand{\Fcomb}{F_{\mathrm{comb}}}
\newcommand{\Fgeom}{F_{\mathrm{geom}}}
\newcommand{\Fgen}{F}

\newcommand{\Hcomb}{H_{\mathrm{comb}}}
\newcommand{\Hgeom}{H_{\mathrm{geom}}}
\newcommand{\Hgen}{H}

\newcommand{\hgeom}{h_{\mathrm{geom}}}

\newcommand{\differentialform}{\Psi}

\newcommand{\Acomb}{\Omega_{\mathrm{comb}}}
\newcommand{\Ageom}{\Omega_{\mathrm{geom}}}
\newcommand{\Agen}{\Omega}

\newcommand{\absmu}{|\mu|}

\newcommand{\qbar}{q}

\newcommand{\Frobeniusbasis}{\psi}

\theoremstyle{plain}
\newtheorem{theoremcounter}{}[]

\newtheorem{proposition}[theoremcounter]{Proposition}

\begin{document}

\thispagestyle{empty}

\begin{flushright}
  MITP/25-052
\end{flushright}

\vspace{1.5cm}

\begin{center}
  {\Large\bf The unequal-mass three-loop banana integral \\
  }
  \vspace{1cm}
  {\large Sebastian~P\"ogel${}^{a}$, 
          Toni~Teschke${}^{b}$, 
          Xing~Wang${}^{c}$ and
          Stefan~Weinzierl${}^{b}$ \\
  \vspace{1cm}
      {\small \em ${}^{a}$ Paul Scherrer Institut, CH-5232 Villigen PSI, Switzerland} \\
  \vspace{2mm}
      {\small \em ${}^{b}$ PRISMA Cluster of Excellence, Institut f{\"u}r Physik, Staudinger Weg 7,} \\
      {\small \em Johannes Gutenberg-Universit{\"a}t Mainz, D - 55099 Mainz, Germany}\\
  \vspace{2mm}
      {\small \em ${}^{c}$ School of Science and Engineering,} \\
      {\small \em The Chinese University of Hong Kong, Shenzhen, 518172  Guangdong, China}
  } 
\end{center}

\vspace{2cm}

\begin{abstract}\noindent
  {
We compute the three-loop banana integral with four unequal masses in dimensional regularisation. 
This integral is associated to a family of K3 surfaces, thus representing an example for Feynman integrals with geometries beyond elliptic curves. 
We evaluate the integral by deriving an $\eps$-factorised differential equation, for which we rely on the algorithm presented in a recent publication~\cite{Bree:2025maw}.
Equipping the space of differential forms in Baikov representation by a set of filtrations inspired by Hodge theory, we first obtain a differential equation with entries as Laurent polynomials in $\eps$.
Via a sequence of basis rotations we then remove any non-$\eps$-factorising terms.
This procedure is algorithmic and at no point relies on prior knowledge of the underlying geometry.
}
\end{abstract}

\vspace*{\fill}

\newpage

\section{Introduction}
\label{sect:intro}

Perturbative Quantum Field Theory (QFT) has proven to be an exceptionally effective tool for obtaining physical observables in particle physics.
A central class of objects in QFT are Feynman integrals, integrals over the momenta of virtual particles participating at a given perturbative level.
These integrals are independent of any specific Lagrangian, and as such encode fundamental mathematical features of QFT.
They are essential for making precision predictions at current and future collider experiments, such as the LHC and FCC.
At the same time, they exhibit deep number-theoretic and geometric structure.
From the perspectives of fundamental physics, particle phenomenology, and mathematics, Feynman integrals thus represent a central object of study.

In dimensional regularisation we evaluate Feynman integrals in space-time dimension $D=D_0-2\eps$, and wish to derive a Laurent series in $\eps$.
The method of differential equations has proven to be one of the most systematic approaches for this task, both for obtaining analytic and numerical results.
Integration-by-parts (IBP) imposes linear relations on the space of Feynman integrals.
The resulting vector space has been shown to be finite~\cite{Smirnov:2010hn}.
A basis of integrals $I$ of this space---usually called master integrals---then satisfies a linear system of differential equations $\diff I=A(\eps,\kinvar)I$, with $\kinvar$ and $\diff$ being the set of kinematic variables and the associated differential, and $A(\eps,\kinvar)$ being a matrix of differential one-forms.
The task of obtaining Laurent series in $\eps$ is trivialised if one is able to find a basis of master integrals $K$, such that the associated differential equation is in $\eps$-factorised form~\cite{Henn:2013pwa}, i.e.~$\diff K=\eps\tilde{A}(\kinvar)K$.
The Laurent series coefficients are then expressible as iterated integrals with kernels determined by $\tilde{A}$.

The complexity of Feynman integrals largely depends on the number of loops and the number of independent kinematic scales.
Simple cases often evaluate to polylogarithms or their generalisations, while the presence of internal masses tends to give rise to more intricate function spaces.
These function spaces are known to admit geometric interpretations, such that it is natural to associate specific geometries to Feynman integrals.
Notably, elliptic curves, hyperelliptic curves, and Calabi–Yau manifolds are known to appear in this context.

Over the years, various methods have been developed to obtain bases with $\varepsilon$-factorised differential equations, with approaches largely determined by the class of underlying geometry. 
General methods for polylogarithmic integrals are well-established~\cite{moserOrderSingularityFuchs1959,Lee:2014ioa,Lee:2017oca} and have been implemented in computational packages~\cite{Prausa:2017ltv,Gituliar:2017vzm,Meyer:2017joq,Lee:2020zfb}. 
For elliptic integrals, techniques are discussed in refs.~\cite{Frellesvig:2021hkr,Dlapa:2022wdu,Gorges:2023zgv,Chaubey:2025adn}, while methods for Calabi--Yau integrals depending on one kinematic parameter are treated in refs.~\cite{Pogel:2022ken,Frellesvig:2024rea,Duhr:2025lbz}. 
Differential equations for hyperelliptic integrals were first addressed in ref.~\cite{Duhr:2024uid}. 
Recently, a generic method independent of the underlying geometry was presented in ref.~\cite{Bree:2025maw}.

Focusing on the class of Calabi--Yau manifolds, the prototypical examples of Feynman integrals with such geometries are the banana integrals.
At $\ell$ loops, these integrals are known to be associated with families of Calabi--Yau $(\ell-1)$-folds~\cite{Klemm:2019dbm,Bonisch:2020qmm,Bonisch:2021yfw,Candelas:2021lkc}, parametrised by the external kinematics.
The associated Calabi--Yau manifolds are also known as Hulek--Verrill manifolds~\cite{Hulek_Verrill_2005,Candelas:2021lkc}.
Consequently, banana integrals have been extensively studied, both at two loops~\cite{Sabry:1962rge,Bogner:2019lfa,Broadhurst:1993mw,Laporta:2004rb,Bloch:2013tra,Adams:2017ejb,Broedel:2017kkb,Broedel:2017siw,Adams:2018yfj,Honemann:2018mrb} and at higher loop orders~\cite{Groote:2005ay,Bloch:2014qca,Primo:2017ipr,Vanhove:2018mto,Broedel:2019kmn,Broedel:2021zij,Kreimer:2022fxm}.
In the equal-mass limit, a method to obtain $\varepsilon$-factorised differential equations at all loop orders was described by a subset of the authors in refs.~\cite{Pogel:2022yat,Pogel:2022ken,Pogel:2022vat}.
The first results for $\eps$-factorised differential equations involving more than one distinct mass were presented in refs.~\cite{Maggio:2025jel,Duhr:2025ppd}.
At three loops, the relevant function spaces are known to exhibit special relations to elliptic modular forms~\cite{verrill1996root,GSJoyce_1972} and their generalisations~\cite{Duhr:2025tdf}.
Despite this progress, the explicit computation of banana integrals with fully generic mass configurations beyond two loops remains an open problem.

In the present paper we take the next step, and present for the first time results for the three-loop banana integral with all four masses being distinct, see fig.~\ref{fig:banana}.
The integral is associated with a Calabi--Yau 2-fold, also known as a K3 surface.
An independent treatment of the same integral appears simultaneously in ref.~\cite{Duhr:threeLoopBanana}.

The challenges in computing this integral are two-fold.
First, having distinct masses leads to large expressions in intermediate steps.
Second, the number of required master integrals increases.
The integral family with four unequal masses involves a system consisting of 15 master integrals: four trivial tadpole integrals and an 11-dimensional non-trivial sector.
This non-trivial sector contains not only integrals directly associated with the K3 geometry, but also an additional set of integrals that arise due to working in non-integer dimensions.

To compute the integral at hand, we derive an $\varepsilon$-factorised differential equation using the algorithm described in ref.~\cite{Bree:2025maw}, for which this paper serves as a highly non-trivial example.
Working in the loop-by-loop Baikov representation~\cite{Frellesvig:2017aai,Frellesvig:2024ymq,Jiang:2024eaj}, we equip the space of differential forms associated with the integral with a set of filtrations inspired by Hodge theory.
These filtrations allow us to separate the sector involving 11 master integrals into six integrals related to the K3 surface and five additional integrals.
The latter are associated with isolated points on the K3 surface where differential forms acquire logarithmic singularities that are regulated in the sense of twisted cohomology~\cite{Mastrolia:2018uzb,Frellesvig:2019uqt}.
Choosing a set of master integrands according to this separation, including a certain pole-order criterion described in ref.~\cite{Bree:2025maw}, automatically leads to a differential equation, where all entries of the matrix $A(\eps,y)$ are Laurent polynomials in $\varepsilon$, with the lowest order being $\varepsilon^{-2}$ and the highest being $\varepsilon^1$.
Furthermore, since the five additional masters are associated with isolated points, they exhibit similarities to polylogarithmic integrals.
Following the methodology of constructing $\mathrm{d}\log$ integrands~\cite{Henn:2020lye,Arkani-Hamed:2017tmz,Arkani-Hamed:2010pyv,Cachazo:2008vp}, it is straightforward to choose a good basis where the associated sub-block of the differential equation is automatically $\varepsilon$-factorised.  
More generally, the method of ref.~\cite{Bree:2025maw} guarantees that the Laurent series coefficients of the differential equation have a simple block-structure induced by one of the filtrations.
Given this structure, we can perform a series of basis rotations that systematically eliminate any remaining terms of order $\varepsilon^{-2}$, $\varepsilon^{-1}$, and $\varepsilon^{0}$.
Thus we are left with terms only of order $\varepsilon^1$ and have derived an $\varepsilon$-factorised differential equation.

The described approach implicitly captures the underlying geometry of the Feynman integral by classifying the space of differential forms in Baikov representation.
These forms are defined in projective space, which may be weighted when the parametrisation introduces a hypersurface that can be realised as a double cover.
This space is equipped with divisors defined by the Baikov polynomials, whose vanishing loci define hypersurfaces that serve as the geometric loci of interest.
Differential forms are then naturally associated with these hypersurfaces or their intersections.
In the example at hand, the K3 surface is realised as a double cover $y^2-P(z_1,z_2)=0$ in weighted projective space, with associated differential forms corresponding to the known cohomology.
Additional cohomology classes are classified by their support on intersections of this K3 surface with divisors
defined by the remaining Baikov polynomials through the presence of logarithmic singularities along these divisors.
Obtaining representatives of these cohomology classes is achieved automatically
through a specific integral ordering during IBP reduction, without requiring careful integrand analysis.
The subsequent step of removing non-$\varepsilon$-factorising terms is similarly algorithmic 
and does not require knowledge of the specific geometry (a K3 surface in this case) in advance.
While a priori knowledge of the geometry and associated data is helpful for intermediate steps, the procedure does not depend on it.

This paper is organised as follows:
In the next section we introduce the unequal-mass three-loop banana integral and the notation used throughout this paper.
In section~\ref{sect:step_1} we construct a basis of master integrals, which leads to a differential equation in Laurent polynomial form.
In section~\ref{sect:step_2} we construct the final basis, which has a differential equation in $\eps$-factorised form.
In section~\ref{sect:mum_point} we analyse the unequal-mass three-loop banana integral in a neighbourhood of a point of maximal unipotent monodromy.
In section~\ref{sect:results} we present results for the unequal-mass three-loop banana integral.
Finally, our conclusions are given in section~\ref{sect:conclusions}.
For the convenience of the reader we give in appendix~\ref{appendix:J_10_J_14} the explicit expressions for the master integrals $J_{10},\dots,J_{14}$ in terms of the basis $I$.
In appendix~\ref{appendix:dgl_step_2} we work out the $\eps$-independent differential equations required in step two 
for the case where there are three non-trivial parts in the filtration.
This is the case which corresponds to the unequal-mass three-loop banana integral.


\section{Notation and setup}
\label{sect:notation}


We are interested in the three-loop banana integral with unequal masses, shown in fig.~\ref{fig:banana_unequal}.
\begin{figure}
\begin{center}
\includegraphics[scale=1.0]{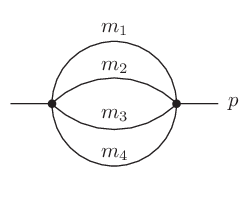}
\end{center}
\caption{
\label{fig:banana}
The three-loop banana integral with unequal masses.
}
\label{fig:banana_unequal}
\end{figure}
To this aim 
we consider the following family of Feynman integrals 
\bq
\label{def_three_loop_banana_integral}
 I_{\nu_1 \nu_2 \nu_3 \nu_4 \nu_5 \nu_6 \nu_7 \nu_8 \nu_9}
 = e^{3 \gamma_E \eps}
 \left(\mu^2\right)^{\nu-\frac{3}{2}D}
 \int \left( \prod\limits_{a=1}^{3} \frac{d^Dk_a}{i \pi^{\frac{D}{2}}} \right)
 \frac{1}{\prod\limits_{b=1}^{9} \sigma_b^{\nu_b}},
\eq
where $D$ denotes the number of space-time dimensions,
$\eps$ the dimensional regularisation parameter,
$\gamma_E$ the Euler--Mascheroni constant
and 
\bq
 \nu & = & \sum\limits_{j=1}^9 \nu_j.
\eq 
The parameter $\mu$ is an arbitrary scale introduced to render the integrals dimensionless.
Without loss of generality we set $\mu^2=-p^2$, where $p$ is the external momentum flowing through the two-point function.
We consider these integrals in $D=2-2\eps$ space-time dimensions.
The inverse propagators are given by 
\begin{align}
 \sigma_1 & = -k_1^2 + m_1^2,
 &
 \sigma_2 & = -k_2^2 + m_2^2,
 &
 \sigma_3 & = -k_3^2 + m_3^2,
 \nonumber \\
 \sigma_4 & = -\left(k_1+k_2+k_3-p\right)^2 + m_4^2,
 &
 \sigma_5 & = -\left(k_1+k_2-p\right)^2,
 &
 \sigma_6 & = -\left(k_1-p\right)^2,
 \nonumber \\
 \sigma_7 & = -\left(k_1+k_2\right)^2,
 & 
 \sigma_8 & = -\left(k_1+k_3\right)^2,
 &
 \sigma_9 & = -\left(k_2+k_3\right)^2.
\end{align}
We are interested in the Feynman integrals with
\bq
 \nu_5, \nu_6, \nu_7, \nu_8, \nu_9 \le 0.
\eq
The variables $\sigma_1, \dots, \sigma_9$ are the Baikov variables.
They can be grouped in three groups:
The variables $\sigma_1, \dots, \sigma_4$ correspond to the propagators of the 
unequal-mass three-loop banana integral.
The variables $\sigma_5$ and $\sigma_6$ are the additional Baikov variables, which appear in a loop-by-loop Baikov representation \cite{Frellesvig:2017aai,Frellesvig:2024ymq}.
In particular, on the maximal cut of the three-loop banana integral
(where we set $\sigma_1=\sigma_2=\sigma_3=\sigma_4=0$) we have a two-dimensional Baikov representation with variables $\sigma_5$ and $\sigma_6$.
The variables $\sigma_7,\sigma_8,\sigma_9$ are only relevant for integration-by-parts reduction programs.

There is some freedom on how to choose the irreducible scalar products.
Our choice is adapted to the loop-by-loop Baikov representation, where we first consider 
the loop formed by the inverse propagators $\sigma_3$ and $\sigma_4$,
followed by 
the loop formed by the inverse propagators $\sigma_2$ and $\sigma_5$,
and finally 
the loop formed by the inverse propagators $\sigma_1$ and $\sigma_6$ 
(this corresponds to the graph of sector $63$ in fig.~\ref{fig:super_sectors}).
On the maximal cut of the three-loop banana integral we obtain a two-dimensional representation (with Baikov variables $\sigma_5$ and $\sigma_6$), 
given by (up to an irrelevant phase)
\bq
\label{def_Baikov_representation}
 \left(2\pi i\right)^4
 \mathrm{Res}_{\sigma_1,\sigma_2,\sigma_3,\sigma_4} \left( I_{1 1 1 1 \nu_5 \nu_6 0 0 0} \right)
 = 
 \frac{2^{6+6 \eps} \pi^{\frac{9}{2}} e^{3 \gamma_E \eps} \left(p^2\right)^{\nu_5+\nu_6+1+4\eps}}{\Gamma\left(\frac{1}{2}-\eps\right)^3}
 \int \frac{d\sigma_5 d\sigma_6}{\left(2\pi i\right)^2} \; u \;
 \sigma_5^{-\nu_5} \sigma_6^{-\nu_6},
 \;
\eq
where the twist function \cite{Mastrolia:2018uzb,Frellesvig:2019uqt} is given by
\bq 
\label{def_u}
 u 
 & = &
 \sigma_5^\eps 
 \sigma_6^\eps 
 \left[ \sigma_6^2 +2 \left(m_1^2+p^2\right) \sigma_6 + \left(m_1^2-p^2\right)^2 \right]^{-\frac{1}{2}-\eps}
 \left[ \left(\sigma_6-\sigma_5\right)^2 +2 m_2^2 \left( \sigma_5+\sigma_6\right) + m_2^4 \right]^{-\frac{1}{2}-\eps}
 \nonumber \\
 & &
 \left[ \sigma_5^2 +2 \left(m_3^2+m_4^2\right) \sigma_5 + \left(m_4^2-m_3^2\right)^2 \right]^{-\frac{1}{2}-\eps}.
\eq
We will see that this representation already contains the essential information to construct a good basis of master integrals.

The unequal-mass three-loop banana integrals depend on four (dimensionless) kinematic variables.
A choice for the kinematic variables is
\bq
 \kinvar_1 \; = \; - \frac{m_1^2}{p^2},
 \;\;\;\;\;\;
 \kinvar_2 \; = \; - \frac{m_2^2}{p^2},
 \;\;\;\;\;\;
 \kinvar_3 \; = \; - \frac{m_3^2}{p^2},
 \;\;\;\;\;\;
 \kinvar_4 \; = \; - \frac{m_4^2}{p^2}.
\eq
For each integral $I_{\nu_1 \nu_2 \nu_3 \nu_4 \nu_5 \nu_6 \nu_7 \nu_8 \nu_9}$ 
we define its sector id by
\bq
\label{def_sector_id}
 \mathrm{id}
 & = & \sum\limits_{j=1}^9 2^{j-1} \Theta\left(\nu_j\right).
\eq
Here, $\Theta(x)$ denotes the Heaviside step function, defined by $\Theta(x)=1$ for $x>0$ and $\Theta(x)=0$ otherwise.
The three-loop banana integral has sector id fifteen, and this is the sector we are mainly interested in.
We call this sector the sector of interest.
Sectors with a lower sector id are trivial, they are either zero or tadpoles.
Although our main interest is the three-loop banana integral (for which we have $\nu_5,\nu_6 \le 0$ in the Baikov representation eq.~(\ref{def_Baikov_representation})),
it will be convenient to consider $\nu_5, \nu_6 \in {\mathbb Z}$.
If $\nu_5 > 0$ or $\nu_6 > 0$ we speak about a super-sector.
There are three relevant super-sectors (sectors $31$, $47$ and $63$).
\begin{figure}
\begin{center}
\includegraphics[scale=1.0]{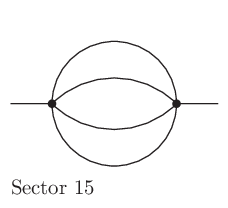}
\includegraphics[scale=1.0]{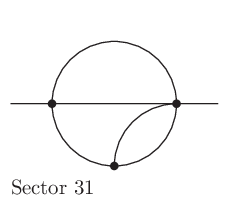}
\includegraphics[scale=1.0]{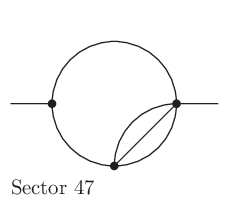}
\includegraphics[scale=1.0]{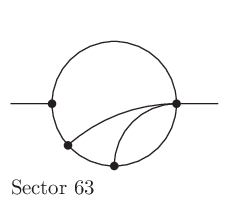}
\end{center}
\caption{
The sector of interest (sector $15$) together with the relevant super-sectors.
}
\label{fig:super_sectors}
\end{figure}
These are shown, together with the sector of interest, in fig.~\ref{fig:super_sectors}.

There are $15$ master integrals in the family of the unequal-mass three-loop banana integrals:
Four master integrals are tadpoles, $11$ master integrals are in the sector of interest.
A possible starting basis $I$ is
\bq
\label{def_basis_I}
 I & = &
 \left(
 I_{111000000}, I_{110100000}, I_{101100000}, I_{011100000},
 \right. \nonumber \\
 & &
 \left.
 \hspace{0.25em}I_{111100000},
 I_{211100000}, I_{121100000}, I_{112100000}, I_{111200000},
 \right. \nonumber \\
 & &
 \left.
 \hspace{0.25em}I_{221100000}, I_{212100000}, I_{211200000},
 I_{122100000}, I_{121200000}, I_{112200000}
 \right)^T.
\eq
This basis is compatible with the six-dimensional loop-by-loop Baikov representation (with the Baikov variables $\sigma_1-\sigma_6$)
of eq.~(\ref{def_Baikov_representation}), as it does not involve the irreducible scalar products $\sigma_7-\sigma_9$.
The super-sectors $31$ and $47$ have one master integral each. A possible basis for these super-sectors is
\bq
 I_{111110000},
 \;\;\;
 I_{111101000}.
\eq
The super-sector $63$ is reducible and does not have any master integral.

The master integrals of the basis $I$ satisfy a differential equation of the form
\bq
 dI\left(\eps,\kinvar\right) & = & A\left(\eps,\kinvar\right) I\left(\eps,\kinvar\right),
\eq
where the entries of the matrix $A(\eps,\kinvar)$ are differential one-forms, rational in $\eps$ and $\kinvar$.
In the main part of the paper we will construct two further bases $J$ and $K$, related to $I$ by
\bq
 J \; = \; R_1^{-1} I, & & K \; = \; R_2^{-1} J \; = \; R_2^{-1} R_1^{-1} I,
\eq
such that the differential equation in the basis $J$ is in Laurent polynomial form (and compatible with a filtration, to be made precise in the next section)
\bq
\label{eq:Laurent}
 d J\left(\eps,\kinvar\right) & = & \sum\limits_{k=-2}^{1} \eps^k \hat{A}^{(k)}\left(\kinvar\right) J\left(\eps,\kinvar\right),
\eq
where $\hat{A}^{(k)}(\kinvar)$ is independent of $\eps$ 
and the differential equation
in the basis $K$ is in $\eps$-factorised form
\bq
 d K\left(\eps,\kinvar\right) & = & \eps \tilde{A}(\kinvar) K\left(\eps,\kinvar\right).
\eq
The matrices $R_1$ and $R_2$ are $(15 \times 15)$-matrices.
The transformation matrix $R_1$ is rational in $\eps$ and in the kinematic variables,
the transformation matrix $R_2$ involves transcendental functions of the kinematic variables.


\section{The construction of the basis $J$}
\label{sect:step_1}

In the first step we construct a basis $J$, which is related to the original basis $I$ by a rotation $R_1(\eps,\kinvar)$.
We follow the general method and the notation of ref.~\cite{Bree:2025maw}.
The essential idea is, that the basis $J$ is obtained directly from the Laporta algorithm if a special order relation is used.
The order relation uses the number of sequential non-zero residues and the pole order of the integrands.
This geometric information can be computed algorithmically from the integrands.
Below we give a detailed discussion how the Laporta algorithm with the ordering relation of ref.~\cite{Bree:2025maw} works in the case
of the unequal-mass three-loop banana integral.

The rotation matrix $R_1(\eps,\kinvar)$ is rational in $\eps$ and in the kinematic variables $\kinvar$.
The basis $J$ has the property that the differential equation in this basis is in Laurent polynomial form, see eq.~(\ref{eq:Laurent}),
and compatible with the $\Fcomb^\bullet$-filtration (to be defined below).
The construction of the basis $J$ is guided by twisted cohomology and Hodge theory.
The sector of interest (sector $15$)
has $11$ master integrals and defines an $11$-dimensional vector space $V^2$.
The superscript ``$2$'' of $V^2$ indicates that the integrands on the maximal cut will be two-forms.
Within twisted cohomology we study the integrands on the maximal cut.
We denote the (infinite-dimensional) vector space of integrands by $\Agen_\omega^2$,
and the finite-dimensional vector space of integrands modulo integration-by-parts relations by $\Hgen^2_\omega$.
In our particular case, the dimension of the vector space $\Hgen^2_\omega$ is thirteen.
The dimension of the vector space $\Hgen^2_\omega$ is larger than the dimension of  the vector space $V^2$, 
because it also includes the super-sectors.
For the case at hand the super-sectors are given by sectors $31$, $47$ and $63$. 
The sector $63$ has no master integrals, the sectors $31$ and $47$ have one master integral each.
This explains the difference in dimensions.
\begin{table}
\begin{center}
\begin{tabular}{|l|l|l|}
 \hline
 vector space & dimension & description \\
 \hline
 $\Agen_\omega^2$ & $\dim \Agen_\omega^2 = \infty$ & Integrands on the maximal cut \\
 $H_\omega^2$ & $\dim H_\omega^2 = 13$ & Integrands on the maximal cut mod linear relations \\
 $V^2$ & $\dim V^2 = 11$ & Integrals on the maximal cut mod linear relations \\
 \hline
\end{tabular}
\caption{
The vector spaces and their dimensions.
\label{tab:vector_spaces}
}
\end{center}
\end{table}
In table~\ref{tab:vector_spaces} we summarise the various vector spaces and their dimensions.
There is an injective map
\bq
 \iota & : & V^2 \rightarrowtail \Hgen^2_{\omega},
\eq
obtained from expressing the maximal cut of the Feynman integral in the Baikov representation.

It is convenient to set $\sigma_5 = - p^2 z_1$ and $\sigma_6 = - p^2 z_2$, such that $z_1$ and $z_2$
are dimensionless variables.
We study the twisted cohomology in projective space ${\mathbb C}{\mathbb P}^2$ with homogeneous coordinates 
$[z_0:z_1:z_2]$, such that we recover in the chart $z_0=1$ the Baikov representation on the maximal cut given
in eq.~(\ref{def_Baikov_representation}).
We denote the twist function on projective space by $U([z_0:z_1:z_2])$, whereas the twist function in the affine chart $z_0=1$ is denoted by $u(z_1,z_2)$.
The latter has been given in eq.~(\ref{def_u}).
The twist function $U$ in homogeneous coordinates is given by
\bq
 U
 & = &
 P_0^{4\eps}
 P_1^{\eps}
 P_2^{\eps}
 P_3^{-\frac{1}{2}-\eps}
 P_4^{-\frac{1}{2}-\eps}
 P_5^{-\frac{1}{2}-\eps},
\eq
with
\begin{align}
 P_0 & = z_0,
 &
 P_3 & = z_2^2 -2 \left(1-\kinvar_1\right) z_0 z_2 + \left(1+\kinvar_1\right)^2 z_0^2,
 \nonumber \\
 P_1 & = z_1,
 &
 P_4 & = \left(z_2-z_1\right)^2 +2 \kinvar_2 z_0 \left( z_1+z_2\right) + \kinvar_2^2 z_0^2,
 \nonumber \\
 P_2 & = z_2,
 &
 P_5 & = z_1^2 +2 \left(\kinvar_3+\kinvar_4\right) z_0 z_1 + \left(\kinvar_3-\kinvar_4\right)^2 z_0^2.
\end{align}
We divide these polynomials into two sets, depending on whether the non-$\eps$ part of the exponent is $0$ or $-\frac{1}{2}$.
We set
\bq
 I^0_{\mathrm{even}} \;= \; \left\{ 0,1,2 \right\},
 & &
 I^0_{\mathrm{odd}} \;= \; \left\{ 3,4,5 \right\}.
\eq
The zero sets of the polynomials $P_0, \dots, P_5$ define hypersurfaces in ${\mathbb C}{\mathbb P}^2$.
The way in which these hypersurfaces intersect plays an important role in the computation of the corresponding Feynman integral.
In general we may have non-normal crossings and indeed we will encounter non-normal crossings in the specific case of the unequal-mass three-loop
banana integral.
In fig.~\ref{fig:twist_zero_set} we show 
\begin{figure}
\begin{center}
\includegraphics[scale=1.2]{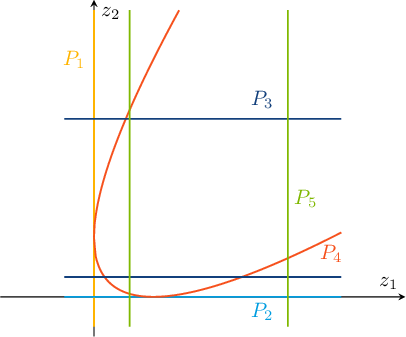}
\end{center}
\caption{
The (real) zero set of the polynomials $P_1, \dots, P_5$ in the chart $z_0=1$.
}
\label{fig:twist_zero_set}
\end{figure}
the (real) zero set of the polynomials $P_1, \dots, P_5$ in the chart $z_0=1$ for $(\kinvar_1,\kinvar_2,\kinvar_3,\kinvar_4)=(-4,-3,-5,-\frac{4}{5})$.
The zero sets of $\{P_1,P_4\}$ and $\{P_2,P_4\}$ have non-normal crossings.
Performing the analysis in all charts, we find that the zero sets of
\bq
 \{P_0,P_3\},
 \{P_0,P_4\},
 \{P_0,P_5\},
 \{P_1,P_4\},
 \{P_1,P_5\},
 \{P_2,P_3\},
 \{P_2,P_4\},
 \{P_0,P_1,P_5\},
 \{P_0,P_2,P_3\}
\eq
have non-normal crossings.

From eq.~(\ref{def_Baikov_representation}) we further read off
\bq
 \prebaikov
 & = &
 \frac{2^{6+6 \eps} \pi^{\frac{9}{2}} e^{3 \gamma_E \eps}}{\Gamma\left(\frac{1}{2}-\eps\right)^3}.
\eq
It is easily checked that with
\bq
 \preabs & = & \eps^3
\eq
the product $\prebaikov \preabs$ is pure of transcendental weight zero. Up to order $\eps^7$ we have
\bq
 \prebaikov \preabs
 & = &
 64 \pi^3 \eps^3 - 48 \pi^5 \eps^5 - 448 \pi^3 \zeta_3 \eps^6 + 10 \pi^7 \eps^7 + {\mathcal O}\left(\eps^8\right).
\eq
We are considering differential two-forms, which can be written as
\bq
\label{def_diff_form}
 \differentialform_{\mu_0 \dots \mu_5}\left[Q\right]
 & = &
 \prebaikov \preabs \prerel \preclutch
 U\left(z\right)
 \hat{\Phi}_{\mu_0 \dots \mu_5}\left[Q\right]
 \eta,
\eq
where
\bq
 \eta & = & z_0 dz_1 \wedge dz_2 - z_1 dz_0 \wedge dz_2 + z_2 dz_0 \wedge dz_1,
\eq
and $\hat{\Phi}_{\mu_0 \dots \mu_5}\left[Q\right]$ is of the form
\bq
 \hat{\Phi}_{\mu_0 \dots \mu_5}\left[Q\right]
 & = &
 \frac{Q\left(z\right)}{P_0^{\mu_0} P_1^{\mu_1} P_2^{\mu_2} P_3^{\mu_3} P_4^{\mu_4} P_5^{\mu_5}},
 \;\;\;\;\;\;
 \mu_j \in {\mathbb N}_0.
\eq
$Q(z)$ is a homogeneous polynomial in $z$ of degree $\mu_0+\mu_1+\mu_2+2(\mu_3+\mu_4+\mu_5)$.
The prefactor $\prerel$ is defined by
\bq
 \prerel
 & = &
 \frac{\Gamma\left(4 \eps +1\right)}{\Gamma\left(4 \eps +1-\mu_0\right)}
 \frac{\Gamma\left(\eps +1\right)}{\Gamma\left(\eps +1-\mu_1\right)}
 \frac{\Gamma\left(\eps +1\right)}{\Gamma\left(\eps +1-\mu_2\right)}
 \prod\limits_{i \in I_{\mathrm{odd}}^0} 
 \frac{\Gamma\left(-\frac{1}{2}-\eps +1\right)}{\Gamma\left(-\frac{1}{2}-\eps +1-\mu_i\right)},
\eq
the prefactor $\preclutch$ is defined by
\bq
\label{def_clutch_prefactor}
 \preclutch
 & = &
 \eps^{-|\mu|},
 \;\;\;\;\;\;
 \left|\mu\right| \; = \; \mu_0 + \mu_1 + \mu_2 + \mu_3 + \mu_4 + \mu_5.
\eq
The prefactors $\prerel$ and $\preclutch$ 
have been meticulously defined to trivialise the $\eps$-dependence of the integration-by-parts identities \cite{Bree:2025maw}.

We denote the vector space spanned by the differential forms of eq.~(\ref{def_diff_form}) by $\Agen^2_\omega$.
The quotient space, where we mod out all relations coming from integration-by-parts, distributivity and cancellations
between numerators and denominators is denoted by $\Hgen^2_{\omega}$.
The vector space $\Agen^2_\omega$ is infinite-dimensional,
the vector space $\Hgen^2_{\omega}$ is finite-dimensional.
To each differential two-form $\differentialform_{\mu_0 \dots \mu_5}[Q] \in \Agen^2_\omega$ we associate three numbers
$|\mu|, o,r \in {\mathbb N}_0$. 
The number $|\mu|$ has been defined above in eq.~(\ref{def_clutch_prefactor}),
the number $o$ denotes the pole order of $\differentialform_{\mu_0 \dots \mu_5}[Q]$, 
the number $r$ denotes the number of non-zero residues, which can be taken successively.
These three numbers define three filtrations $W_\bullet$, $\Fgeom^\bullet$ and $\Fcomb^\bullet$.
We have
\begin{alignat}{2}
 \differentialform_{\mu_0,\dots,\mu_5}[Q] & \in W_w \Agen^2_\omega & \quad \mbox{if} & \quad 2 + r \le w,
 \nonumber \\
 \differentialform_{\mu_0,\dots,\mu_5}[Q] & \in \Fgeom^{p} \Agen^2_\omega & \quad \mbox{if} & \quad 2+r-o \ge p,
 \nonumber \\
 \differentialform_{\mu_0,\dots,\mu_5}[Q] & \in \Fcomb^{p'} \Agen^2_\omega & \quad \mbox{if} & \quad 2-|\mu| \ge p'.
\end{alignat}
The filtration $W_\bullet$ is increasing
\bq
 \ldots \subseteq W_2 \Agen^2_\omega \subseteq W_3 \Agen^2_\omega \subseteq W_4 \Agen^2_\omega \subseteq \dots,
\eq
the filtrations $\Fgeom^\bullet$ and $\Fcomb^\bullet$ are decreasing
\bq
 \ldots \subseteq \Fgeom^{2} \Agen^2_\omega \subseteq \Fgeom^{1} \Agen^2_\omega \subseteq \Fgeom^{0} \Agen^2_\omega \subseteq \dots,
 \nonumber \\
 \ldots \subseteq \Fcomb^{2} \Agen^2_\omega \subseteq \Fcomb^{1} \Agen^2_\omega \subseteq \Fcomb^{0} \Agen^2_\omega \subseteq \dots.
\eq
We set
\bq
 \Ageom^{p,q} = \mathrm{Gr}^{p}_{\Fgeom} \mathrm{Gr}^W_{p+q} \Agen^2_\omega
 & \mbox{and} &
 \Acomb^{p',q'} = \mathrm{Gr}^{p'}_{\Fcomb} \mathrm{Gr}^W_{p'+q'} \Agen^2_\omega,
\eq
where the graded parts are defined by
\bq
 \mathrm{Gr}^W_{w} X = W_{w} X / W_{w-1} X
 & \mbox{and} &
 \mathrm{Gr}^{p}_{F} X = F^p X / F^{p+1} X.
\eq
We let $\Hgeom^{p,q}$ be generated by all master integrands $\differentialform \in \Ageom^{p,q}$
and $\Hcomb^{p',q'}$ be generated by all master integrands $\differentialform \in \Acomb^{p',q'}$.
We further set $\hgeom^{p,q}=\dim \Hgeom^{p,q}$.
The Feynman integrals on the maximal cut inherit this labelling.
We denote by $V_{\mathrm{geom}}^{p,q}$ the vector space generated by the master integrals, whose integrands belong
to $\Hgeom^{p,q}$, and by
$V_{\mathrm{comb}}^{p,q}$ the vector space generated by the master integrals, whose integrands belong
to $\Hcomb^{p,q}$.

\newpage
The $W_\bullet$-filtration and the $\Fgeom^\bullet$-filtration decompose the $13$-dimensional vector space $\Hgen^2_\omega$ into smaller pieces and we find
\bq
\begin{picture}(120,30)(0,30)
\Text(10,10)[c]{$\hgeom^{2,0}$}
\Text(50,10)[c]{$\hgeom^{1,1}$}
\Text(90,10)[c]{$\hgeom^{0,2}$}
\Text(30,30)[c]{$\hgeom^{2,1}$}
\Text(70,30)[c]{$\hgeom^{1,2}$}
\Text(50,50)[c]{$\hgeom^{2,2}$}
\end{picture}
 & = &
\begin{picture}(120,30)(-20,30)
\Text(10,10)[c]{$1$}
\Text(50,10)[c]{$4$}
\Text(90,10)[c]{$1$}
\Text(30,30)[c]{$0$}
\Text(70,30)[c]{$0$}
\Text(50,50)[c]{$7$}
\end{picture}
 \\
 & & \nonumber \\
 & & \nonumber
\eq
We have $\hgeom^{2,2}=7$, but note this includes the two super-sectors $31$ and $47$ with one master integral each.
Thus we expect five master integrals from sector $15$ in $\Hgeom^{2,2}$ 
and the vector space $V^2$ of Feynman integrals on the maximal cut decomposes as
\bq
\begin{picture}(120,30)(-20,30)
\Text(10,10)[c]{$1$}
\Text(50,10)[c]{$4$}
\Text(90,10)[c]{$1$}
\Text(30,30)[c]{$0$}
\Text(70,30)[c]{$0$}
\Text(50,50)[c]{$5$}
\end{picture}
 \\
 & & \nonumber \\
 & & \nonumber
\eq
In the Laporta algorithm we use the sequence of integers
\bq
 \laportaorder,
\eq
where the dots stand for further criteria needed to distinguish inequivalent integrands.
The sequence of integers is used with lexicographic ordering,
i.e. the relation $a_1 < a_2$ implies $\differentialform_1 < \differentialform_2$, with ties broken by $w$, etc.. 
The entries in this sequence of integers are defined as follows:
We first consider (recursively) all sub-problems localised on $\Divisor_i=0$ with $i \in I_{\mathrm{even}}^0$.
Localisation means to take a residue at $\Divisor_i=0$.
These are simpler problems, as the number of Baikov variables is reduced by one.
We set
\bq
 a & = & \left\{ \begin{array}{rl}
 -r & \mbox{if $\differentialform$ is the pre-image of a master integrand of a sub-problem} \\
    & \mbox{localised on $\Divisor_i=0$ with $i \in I_{\mathrm{even}}^0$,} \\
 0 & \mbox{otherwise.} \\
 \end{array}
 \right.
\eq
For a Baikov representation with two variables the (Hodge) weight is defined by $w=2+r$.
The integers $r$, $o$ and $|\mu|$ denote respectively the number of sequential non-zero residues, the pole order and the sum of the indices $\mu_i$
and have been defined before.

We denote by $\differentialform=(\differentialform_5,\dots,\differentialform_{17})^T$
the basis of master integrands obtained from this algorithm.
This basis satisfies a differential equation as in eq.~(\ref{eq:Laurent}).
Moreover, if $\differentialform_i \in \mathrm{Gr}_{\Fcomb}^{2-\absmu_i} \Agen^2_\omega$
and $\differentialform_j \in \mathrm{Gr}_{\Fcomb}^{2-\absmu_j} \Agen^2_\omega$, then
\bq
\label{refined_statement}
 A_{ij}\left(\eps,x\right)
 & = &
 \sum\limits_{k=-(\absmu_i-\absmu_j)}^1
 \eps^k A^{(k)}_{ij}\left(x\right).
\eq
We call a differential equation which satisfies eq.~(\ref{refined_statement}) an $\Fgen^\bullet$-compatible differential equation for the filtration $\Fcomb^\bullet$.

\subsection{Weight four}

We start with differential forms, where we may take two consecutive residues.
These differential forms are of weight four and belong to $\Hgeom^{2,2}$.
In order to take a residue, one of the indices $\mu_0$, $\mu_1$ or $\mu_2$ of the even polynomials needs to be positive.
It is not too difficult to see that this will result in a differential two-form
where another residue can be taken, as one of the odd polynomials will become a perfect square in the limit where the even polynomial
vanishes.

Let us write down all candidates for $\hat{\Phi}$ with $|\mu|\le 2$:
\bq
 \hat{\Phi}
 & \in &
 \left\{
  \frac{z_1}{z_0},
  \frac{z_2}{z_0},
  \frac{z_0}{z_1},
  \frac{z_2}{z_1},
  \frac{z_0}{z_2},
  \frac{z_1}{z_2},
  \frac{z_2^2}{z_0 z_1},
  \frac{z_1^2}{z_0 z_2},
  \frac{z_0^2}{z_1 z_2}
 \right\}.
\eq
With respect to the $\Fcomb$-filtration, the first six candidates have $|\mu|=1$ and belong to $\Fcomb^1 \Agen^2_\omega$, 
while  the last three candidates have $|\mu|=2$ and belong to $\Fcomb^0 \Agen^2_\omega$.
Out of these nine possibilities only seven are independent.
The order relation will select master integrands with lowest pole order.
With two residues, we will need at least pole order two.
As $\{P_0,P_1,P_5\}$ and $\{P_0,P_2,P_3\}$ have non-normal crossing singularities,
the candidates $z_2^2/(z_0z_1)$ and $z_1^2/(z_0z_2)$ have pole order higher than two.
The order relation will eliminate those in favour of the other ones.
Thus $\Hgeom^{(2,2)}$ is generated by
\bq
 \left\{
 \differentialform_{100000}\left[z_1\right],
 \differentialform_{100000}\left[z_2\right],
 \differentialform_{010000}\left[z_0\right],
 \differentialform_{010000}\left[z_2\right],
 \differentialform_{001000}\left[z_0\right],
 \differentialform_{001000}\left[z_1\right],
 \differentialform_{011000}\left[z_0^2\right]
 \right\}.
\eq
The last five entries correspond to integrals in super-sectors (sectors $31$, $47$ or $63$).
The sector $63$ has no master integrals, the sectors $31$ and $47$ have one master integral each.
We take $\differentialform_{010000}\left[z_0\right]$ as a representative for the master integral in sector $31$
and $\differentialform_{001000}\left[z_0\right]$ as a representative for the master integral in sector $47$.
Out of the last five entries we can form three independent linear combinations such that these linear combinations
vanish on the Feynman integral side on all super-sectors.
These linear combinations are
\bq
 \differentialform_{010000}\left[z_2+\kinvar_2z_0\right],
 \;\;
 \differentialform_{001000}\left[z_1+\kinvar_2z_0\right],
 \;\;
 \differentialform_{011000}\left[\kinvar_2z_0^2+z_0\left(z_1+z_2\right)\right].
\eq
Projecting these linear combinations to sector $15$ gives us three basis vectors of $V_{\mathrm{geom}}^{2,2}$.
Together with the first two candidates ($\differentialform_{100000}\left[z_1\right]$ and $\differentialform_{100000}\left[z_2\right]$) we obtain in total five basis vectors, as advertised.
The projection onto sector $15$ depends on the choice of a basis of master integrals.

\subsection{Weight three}

There are no master integrands at weight three (i.e. differential forms with just one non-vanishing residue).
The reason was already given above: If either $\mu_0$, $\mu_1$ or $\mu_2$ is positive,
one of the odd polynomials will become a perfect square in the limit where the even polynomial
vanishes and we may take a second residue.

\subsection{Weight two}

The master integrands at weight two do not have non-vanishing residues. We order them by the pole order.

\subsubsection{Pole order zero}

At pole order zero we have $\mu_0=\dots=\mu_5=0$. The polynomial $Q$ is then of degree $0$.
Thus the only choice is
\bq
 \hat{\Phi} & = & 1.
\eq
In this case we have $\prerel=1$ and $\preclutch=1$.
Therefore
\bq
 \differentialform_{000000}\left[1\right]
 & = &
 \prebaikov \eps^3
 U\left(z\right)
 \eta,
\eq
and
\bq
 \iota\left( \eps^3 I_{111100000} \right) & = & \differentialform_{000000}\left[1\right].
\eq

\subsubsection{Pole order one}

At pole order one (and zero residues), we have to increase one of the indices $\mu_3$ $\mu_4$ or $\mu_5$ to one.
Then $\hat{\Phi}$ is of the form
\bq
 \hat{\Phi}
 & = &
 \frac{Q_i}{P_i},
 \;\;\;\;\;\; i \in I^0_{\mathrm{odd}},
\eq
where $Q_i$ is a homogeneous polynomial of degree two.
Integration-by-parts relations provide relations among these possibilities and we find that four of them are independent.
We make the choice
\bq
 \hat{\Phi}
 & \in &
 \left\{
  \frac{2\kinvar_1z_0\left[z_2+\left(1+\kinvar_1\right)z_0\right]}{P_3},
  \frac{2\kinvar_2z_0\left[\left(z_1+z_2\right)+\kinvar_2z_0\right]}{P_4},
 \right.
 \nonumber \\
 & &\hspace{0.8em} 
 \left.
  \frac{2\kinvar_3z_0\left[z_1+\left(\kinvar_3-\kinvar_4\right)z_0\right]}{P_5},
  \frac{2\kinvar_4z_0\left[z_1-\left(\kinvar_3-\kinvar_4\right)z_0\right]}{P_5}
 \right\}.
\eq
For all four cases we have $\prerel=-\frac{1}{2}-\eps$ and $\preclutch=\frac{1}{\eps}$.
We have chosen these forms such that
\bq
 \differentialform_{000100}\left[2\kinvar_1z_0\left(z_2+\left(1+\kinvar_1\right)z_0\right)\right]
 & = & 
 \frac{\kinvar_1}{\eps}
 \frac{\partial}{\partial \kinvar_1} \differentialform_{000000}\left[1\right],
 \nonumber \\
 \differentialform_{000010}\left[2\kinvar_2z_0\left(\left(z_1+z_2\right)+\kinvar_2z_0\right)\right]
 & = & 
 \frac{\kinvar_2}{\eps}
 \frac{\partial}{\partial \kinvar_2} \differentialform_{000000}\left[1\right],
 \nonumber \\
 \differentialform_{000001}\left[2\kinvar_3z_0\left(z_1+\left(\kinvar_3-\kinvar_4\right)z_0\right)\right]
 & = & 
 \frac{\kinvar_3}{\eps}
 \frac{\partial}{\partial \kinvar_3} \differentialform_{000000}\left[1\right],
 \nonumber \\
 \differentialform_{000001}\left[2\kinvar_4z_0\left(z_1-\left(\kinvar_3-\kinvar_4\right)z_0\right)\right]
 & = & 
 \frac{\kinvar_4}{\eps}
 \frac{\partial}{\partial \kinvar_4} \differentialform_{000000}\left[1\right].
\eq

\subsubsection{Pole order two}

It remains to find a generator for $H^{0,2}_\omega$. This space is one-dimensional and generated by the second derivatives.
We take the symmetric combination
\bq
 \frac{1}{16 \eps^2}
 \left( \sum\limits_{i=1}^4 \kinvar_i \frac{\partial}{\partial \kinvar_i} \right)^2 \differentialform_{000000}\left[1\right].
\eq

\subsection{A basis for the twisted cohomology group}

Let us summarise the generators for the twisted cohomology group $\Hgen_\omega^2$.
We have $\dim \Hgen_\omega^2 = 13$.
We label the generators $\differentialform_5, \dots, \differentialform_{17}$, to ease the translation to the basis
$J$, where we reserve the first four master integrals $J_1,\dots,J_4$ for the tadpoles.
A possible basis of $\Hgen_\omega^2$ is given by
\begin{align}
 \differentialform_5 & = \differentialform_{000000}\left[1\right],
 &
 \differentialform_{10} & = \differentialform_{100000}\left[z_1\right],
 \nonumber \\
 \differentialform_6 & = \frac{\kinvar_1}{\eps} \frac{\partial}{\partial \kinvar_1} \differentialform_{000000}\left[1\right],
 &
 \differentialform_{11} & = \differentialform_{100000}\left[z_2\right],
 \nonumber \\
 \differentialform_7 & = \frac{\kinvar_2}{\eps} \frac{\partial}{\partial \kinvar_2} \differentialform_{000000}\left[1\right],
 &
 \differentialform_{12} & = \differentialform_{010000}\left[z_2+\kinvar_2z_0\right],
 \nonumber \\
 \differentialform_8 & = \frac{\kinvar_3}{\eps} \frac{\partial}{\partial \kinvar_3} \differentialform_{000000}\left[1\right],
 &
 \differentialform_{13} & = \differentialform_{001000}\left[z_1+\kinvar_2z_0\right],
 \nonumber \\
 \differentialform_9 & = \frac{\kinvar_4}{\eps} \frac{\partial}{\partial \kinvar_4} \differentialform_{000000}\left[1\right],
 &
 \differentialform_{14} & = \differentialform_{011000}\left[\kinvar_2z_0^2+z_0\left(z_1+z_2\right)\right],
 \nonumber \\
 \differentialform_{15} & =  \frac{1}{16 \eps^2}
 \left( \sum\limits_{i=1}^4 \kinvar_i \frac{\partial}{\partial \kinvar_i} \right)^2 \differentialform_{000000}\left[1\right],
 &
 \differentialform_{16} & = \differentialform_{010000}\left[z_0\right],
 \nonumber \\
 & &
 \differentialform_{17} & = \differentialform_{001000}\left[z_0\right].
\end{align}

\subsection{The master integrals of the basis $J$}

The basis $J$ consists of four master integrals ($J_1,\dots,J_4$) corresponding to the tadpoles and eleven master integrals spanning the vector space $V^2$.
For each of the thirteen integrands $\differentialform_5, \dots, \differentialform_{17}$ we obtain a Feynman integral by integration over the Baikov integration domain.
We are not interested in the Feynman integrals corresponding to $\differentialform_{16}$ and $\differentialform_{17}$, these translate to the integrals $I_{111110000}$ and $I_{111101000}$
and correspond to the master integrals of the super-sectors $31$ and $47$, respectively.
A direct translation of the integrands $\differentialform_{12}-\differentialform_{14}$ to Feynman integrals will give Feynman integrals, 
which live in the super-sectors $31$, $47$ and $63$.
However, we may always subtract a linear combination of $I_{111110000}$ and $I_{111101000}$ such that they vanish on the super-sectors $31$, $47$ and $63$.

The master integrals $J_{10},\dots,J_{14}$ correspond to integrands, where we may take two consecutive residues.
Normalising them by their leading singularity will directly give an $\eps$-factorised form in the $5\times 5$-subsystem.

In summary, we may take the intermediate basis $J$ as
\bq
 J_1
 & = & 
 \eps^3 I_{011100000},
 \nonumber \\
 J_2
 & = & 
 \eps^3 I_{101100000},
 \nonumber \\
 J_3
 & = & 
 \eps^3 I_{110100000},
 \nonumber \\
 J_4
 & = & 
 \eps^3 I_{111000000},
 \nonumber \\
 J_5
 & = & 
 \eps^3 \; I_{111100000}
 \nonumber \\
 J_6
 & = & 
 \frac{1}{\eps} \kinvar_1 \frac{\partial}{\partial \kinvar_1} J_5,
 \nonumber \\
 J_7
 & = & 
 \frac{1}{\eps} \kinvar_2 \frac{\partial}{\partial \kinvar_2} J_5,
 \nonumber \\
 J_8
 & = & 
 \frac{1}{\eps} \kinvar_3 \frac{\partial}{\partial \kinvar_3} J_5,
 \nonumber \\
 J_9
 & = & 
 \frac{1}{\eps} \kinvar_4 \frac{\partial}{\partial \kinvar_4}  J_5,
 \nonumber \\
 J_{10}
 & = & 
 \left.
 \eps^3 \left[ 6 I_{1111(-1)0000} - 8 I_{11110(-1)000} + 3 I_{111100000} \right] \right|_{\mathrm{sector} \; 15},
 \nonumber \\
 J_{11}
 & = & 
 \left.
 \eps^3 \left[3 I_{1111(-1)0000} - I_{11110(-1)000} + 3 \kinvar_2 I_{111100000}
 \right] \right|_{\mathrm{sector} \; 15},
 \nonumber \\
 J_{12}
 & = & 
 \eps^3 \left( 1+\kinvar_1 \right) 
 \left. \left( 3 I_{1111(-1)1000} + 3 \kinvar_2 I_{111101000} - I_{111100000} \right)\right|_{\mathrm{sector} \; 15},
 \nonumber \\
 J_{13}
 & = & 
 \eps^3 \left( \kinvar_3-\kinvar_4 \right) \left. \left( I_{11111(-1)000} + \kinvar_2 I_{111110000} \right)\right|_{\mathrm{sector} \; 15},
 \nonumber \\
 J_{14}
 & = & 
 \eps^3 \left( 1+\kinvar_1 \right) \left( \kinvar_3-\kinvar_4 \right) \left. \left( \kinvar_2 I_{111111000} + I_{111110000} + I_{111101000} \right)\right|_{\mathrm{sector} \; 15},
 \nonumber \\
 J_{15}
 & = & 
 \frac{1}{16 \eps^2} \left[ \kinvar_1 \frac{\partial}{\partial \kinvar_1} + \kinvar_2 \frac{\partial}{\partial \kinvar_2} + \kinvar_3 \frac{\partial}{\partial \kinvar_3} + \kinvar_4 \frac{\partial}{\partial \kinvar_4} \right]^2 J_5.
\eq
For the master integrals $J_{10},\dots,J_{14}$ we have taken
appropriate linear combination, such that 
in the $m_3^2=m_4^2$-limit the master integrals $J_{13}$ and $J_{14}$ go to zero.
In the $m_2^2=m_3^2=m_4^2$-limit in addition the master integrals $J_{11}$ and $J_{12}$ go to zero.
Finally, in the equal mass limit ($m_1^2=m_2^2=m_3^2=m_4^2$) in addition 
the master integral $J_{10}$ goes to zero.
This trivialises the boundary condition in the equal mass limit for the master integrals $J_{10},\dots,J_{14}$.

We recall that the master integrals $J_{12},\dots,J_{14}$ are a priori integrals in super-sectors 
(sectors $31$, $47$ and $63$, respectively). 
However, they are constructed such that they vanish on the maximal cuts of the super-sectors.
In the definition of $J_{10},\dots,J_{14}$ we take the projection to the sector $15$,
e.g. we set to zero any integral for which there is a $\nu_j \le 0$ with $j \in \{1,2,3,4\}$.
In particular we project the tadpoles $J_1,\dots,J_4$ to zero.
As mentioned above, the projection onto sector $15$ depends on the choice of master integrals in the basis $I$.
The explicit expressions for $J_{10},\dots,J_{14}$ in terms of the basis $I$ are given in appendix~\ref{appendix:J_10_J_14}.
We remark that although for the master integrals $J_{10}$ and $J_{11}$ the projection to the sector $15$ is not strictly necessary, it
is extremely convenient to do so:
We observe that this eliminates in the differential equation any non-$\eps^1$-contributions from the tadpoles to the top sector.

The differential equation for $J$ has the form
\bq
 d J\left(\eps,\kinvar\right) & = & \hat{A}\left(\eps,\kinvar\right) J\left(\eps,\kinvar\right) \; = \; \sum\limits_{k=-2}^{1} \eps^k \hat{A}^{(k)}\left(\kinvar\right) J\left(\eps,\kinvar\right),
\eq
where $\hat{A}^{(k)}(\kinvar)$ is independent of $\eps$.
Moreover, the entries on the maximal cut are compatible with the $\Fcomb^\bullet$-filtration.
We have
\begin{align}
 J_{5} & \in \Fcomb^2 \Agen^2_\omega,
 &
 J_{6}, \dots, J_{13} & \in \Fcomb^1 \Agen^2_\omega,
 &
 J_{14}, J_{15} & \in \Fcomb^0 \Agen^2_\omega.
\end{align}
In eq.~(\ref{A_ldegree}) we show the lowest non-vanishing order in $\eps$ for each entry
of the connection matrix.
\bq
\label{A_ldegree}
 \left(
 \begin{array}{rrrr|r|rrrrrrrr|rr}
 1 & - & - & - & - & - & - & - & - & - & - & - & - & - & - \\
 - & 1 & - & - & - & - & - & - & - & - & - & - & - & - & - \\
 - & - & 1 & - & - & - & - & - & - & - & - & - & - & - & - \\
 - & - & - & 1 & - & - & - & - & - & - & - & - & - & - & - \\
 \hline
 - & - & - & - & - & 1 & 1 & 1 & 1 & - & - & - & - & - & - \\
 \hline
 - & - & - & - & -1 & 0 & 0 & 0 & 0 & 1 & 1 & 1 & 1 & 1 & 1 \\
 - & - & - & - & -1 & 0 & 0 & 0 & 0 & 1 & 1 & 1 & 1 & 1 & 1 \\
 - & - & - & - & -1 & 0 & 0 & 0 & 0 & 1 & 1 & 1 & 1 & 1 & 1 \\
 - & - & - & - & -1 & 0 & 0 & 0 & 0 & 1 & 1 & 1 & 1 & 1 & 1 \\
 1 & 1 & 1 & 1 & 0 & 1 & 1 & 1 & 1 & 1 & 1 & 1 & 1 & - & - \\
 1 & 1 & 1 & 1 & 0 & 1 & 1 & 1 & 1 & 1 & 1 & 1 & 1 & - & - \\
 1 & 1 & 1 & 1 & 0 & 1 & 1 & 1 & 1 & - & 1 & 1 & - & 1 & - \\
 1 & 1 & 1 & 1 & 0 & 1 & 1 & 1 & 1 & 1 & 1 & - & 1 & 1 & - \\
 \hline
 1 & - & 1 & 1 & 0 & 1 & 1 & 1 & 1 & - & - & 1 & 1 & 1 & - \\
 1 & 1 & 1 & 1 & -2 & -1 & -1 & -1 & -1 & 0 & 0 & 0 & 0 & 0 & 0 \\
 \end{array}
 \right).
\eq
The symbol ``$-$'' indicates that the corresponding entry is zero.
In eq.~(\ref{A_ldegree}) we also indicated the block structure due to the $\Fcomb^\bullet$-filtration.
As a side we note that
the $4 \times 4$-block corresponding to $J_1,\dots,J_4$
and the $5 \times 5$-block corresponding to $J_{10},\dots,J_{14}$ are already in $\eps$-factorised form.


\section{The construction of the basis $K$}
\label{sect:step_2}

In a second step we go from the basis $J$, which has an $\Fcomb^\bullet$-compatible differential equation, to a basis $K$ with an $\eps$-factorised differential equation.
To this aim we construct a rotation matrix $R_2$ such that the differential equation for
the basis $K = R_2^{-1} J$ is in $\eps$-factorised form:
\bq
 d K\left(\eps,\kinvar\right) & = & \eps \tilde{A}(\kinvar) K\left(\eps,\kinvar\right).
\eq
The rotation matrix $R_2$ is given as the product of simpler matrices
\bq
 R_2 & = & R_2^{(-2)} R_2^{(-1)} R_2^{(0)}.
\eq
The product $R_2^{(-2)} R_2^{(-1)} R_2^{(0)}$ puts the top sector into an $\eps$-factorised form.
The individual rotations $R_2^{(-2)}$, $R_2^{(-1)}$ and $R_2^{(0)}$
eliminate systematically terms of order $\eps^{-2}$, $\eps^{-1}$ and $\eps^{0}$.
The matrix $\hat{A}$ restricted to the maximal cut has a block structure
induced by the $\Fcomb^\bullet$-filtration and indicated by lines in eq.~(\ref{A_ldegree}).
We define a $B$-order by picking in each block terms of a certain order in $\eps$ as follows:
\begin{align}
 &\mbox{$B$-order}\ -2
 :
 \scalebox{0.9}{$\left( \begin{array}{ccc}
 0 & - & - \\
 -1 & 0 & - \\
 -2 & -1 & 0 \\
 \end{array} \right)$},
 & 
 &\mbox{$B$-order}\ -1 
 :
  \scalebox{0.9}{$\left( \begin{array}{ccc}
 - & - & - \\
 0 & - & - \\
 -1 & 0 & - \\
 \end{array} \right)$},
 \nonumber 
 & \mbox{$B$-order}\ 0 
 : 
 \;\;\;\;
  \scalebox{0.9}{$\left( \begin{array}{ccc}
 - & - & - \\
 - & - & - \\
 0 & - & - \\
 \end{array} \right)$}.
\end{align}
The rotation $R_2^{(-2)}$ removes terms of $B$-order $(-2)$, 
the rotation $R_2^{(-1)}$ removes terms of $B$-order $(-1)$ 
and the rotation $R_2^{(0)}$ removes terms of $B$-order $0$.

For each of these rotations we make an ansatz with $\eps$-independent functions 
of the kinematic variables $\kinvar$.
Requiring that these rotations remove the unwanted terms gives us a system of first-order $\eps$-independent 
differential equations for the unknown functions of the ansatz.
For an $\Fcomb^\bullet$-compatible differential equation
with three non-trivial parts in the filtration
\bq
 \emptyset = \Fcomb^3 V^2 \subseteq \Fcomb^2 V^2 \subseteq \Fcomb^1 V^2 \subseteq \Fcomb^0 V^2 = V^2
\eq
we work out these differential equations in appendix~\ref{appendix:dgl_step_2}.
Solving these differential equations gives us the functions of the ansatz and hence the rotation matrices
$R_2^{(j)}$.
It is worth noting that these differential equations have a triangular structure, e.g. we may start 
with solving a smaller closed subsystem involving fewer functions and gradually work our way up.
In the next subsections we discuss the system of  differential equations for the case of the unequal-mass three-loop banana integral.

Note that there is no need to adjust the tadpole contributions to the top sector, as these contributions are already of order $\eps^1$ and this is preserved under
the above transformation.

\subsection{The transformation $R_2^{(-2)}$}

The transformation $R_2^{(-2)}$ removes terms of $B$-order $(-2)$.
The ansatz for $R_2^{(-2)}$ reads
{\footnotesize
\bq
\label{Rm2_ldegree}
 \left(
 \begin{array}{rrrr|r|rrrrrrrr|rr}
 1 & 0 & 0 & 0 & 0 & 0 & 0 & 0 & 0 & 0 & 0 & 0 & 0 & 0 & 0 \\
 0 & 1 & 0 & 0 & 0 & 0 & 0 & 0 & 0 & 0 & 0 & 0 & 0 & 0 & 0 \\
 0 & 0 & 1 & 0 & 0 & 0 & 0 & 0 & 0 & 0 & 0 & 0 & 0 & 0 & 0 \\
 0 & 0 & 0 & 1 & 0 & 0 & 0 & 0 & 0 & 0 & 0 & 0 & 0 & 0 & 0 \\
 \hline
 0 & 0 & 0 & 0 & R^{(-2)}_{55} & 0 & 0 & 0 & 0 & 0 & 0 & 0 & 0 & 0 & 0 \\
 \hline
 0 & 0 & 0 & 0 & \frac{1}{\eps} R^{(-2)}_{65} & R^{(-2)}_{66} & R^{(-2)}_{67} & R^{(-2)}_{68} & R^{(-2)}_{69} & 0 & 0 & 0 & 0 & 0 & 0 \\
 0 & 0 & 0 & 0 & \frac{1}{\eps} R^{(-2)}_{75} & R^{(-2)}_{76} & R^{(-2)}_{77} & R^{(-2)}_{78} & R^{(-2)}_{79} & 0 & 0 & 0 & 0 & 0 & 0 \\
 0 & 0 & 0 & 0 & \frac{1}{\eps} R^{(-2)}_{85} & R^{(-2)}_{86} & R^{(-2)}_{87} & R^{(-2)}_{88} & R^{(-2)}_{89} & 0 & 0 & 0 & 0 & 0 & 0 \\
 0 & 0 & 0 & 0 & \frac{1}{\eps} R^{(-2)}_{95} & R^{(-2)}_{96} & R^{(-2)}_{97} & R^{(-2)}_{98} & R^{(-2)}_{99} & 0 & 0 & 0 & 0 & 0 & 0 \\
 0 & 0 & 0 & 0 & 0 & 0 & 0 & 0 & 0 & 1 & 0 & 0 & 0 & 0 & 0 \\
 0 & 0 & 0 & 0 & 0 & 0 & 0 & 0 & 0 & 0 & 1 & 0 & 0 & 0 & 0 \\
 0 & 0 & 0 & 0 & 0 & 0 & 0 & 0 & 0 & 0 & 0 & 1 & 0 & 0 & 0 \\
 0 & 0 & 0 & 0 & 0 & 0 & 0 & 0 & 0 & 0 & 0 & 0 & 1 & 0 & 0 \\
 \hline
 0 & 0 & 0 & 0 & 0 & 0 & 0 & 0 & 0 & 0 & 0 & 0 & 0 & 1 & 0 \\
 0 & 0 & 0 & 0 & \frac{1}{\eps^2} R^{(-2)}_{F5} & \frac{1}{\eps} R^{(-2)}_{F6} & \frac{1}{\eps} R^{(-2)}_{F7}  & \frac{1}{\eps} R^{(-2)}_{F8}  & \frac{1}{\eps} R^{(-2)}_{F9}  & 0 & 0 & 0 & 0 & R^{(-2)}_{FE} & R^{(-2)}_{FF} \\
 \end{array}
 \right).
\eq
}
Here we used hexadecimal notation for the indices and we indicated the block structure 
due to the $\Fcomb^\bullet$-filtration.
It is convenient to rename the entry $R^{(-2)}_{55}$:
\bq
 R^{(-2)}_{55} & = & \Frobeniusbasis.
\eq
A posteriori we can show that $\Frobeniusbasis$ is a period of the K3-surface and this motivates the notation.
However, we stress that we never need this information.

Requiring that the terms of order $\eps^0$ of the entry $(5,5)$ in the transformed differential equation vanish yields
\begin{align}
\label{def_Rm2_65_to_75}
 R^{(-2)}_{65} & = \kinvar_1 \frac{\partial}{\partial \kinvar_1} \Frobeniusbasis,
 &
 R^{(-2)}_{75} & = \kinvar_2 \frac{\partial}{\partial \kinvar_2} \Frobeniusbasis,
 &
 R^{(-2)}_{85} & = \kinvar_3 \frac{\partial}{\partial \kinvar_3} \Frobeniusbasis,
 &
 R^{(-2)}_{95} & = \kinvar_4 \frac{\partial}{\partial \kinvar_4} \Frobeniusbasis.
\end{align}
In the next step we look at the terms of order $\eps^{-1}$ of the entries $(6,5), (7,5), (8,5), (9,5)$.
This allows us to express all second derivatives of $\Frobeniusbasis$ in terms of $\Frobeniusbasis$,
first derivatives of $\Frobeniusbasis$ and $R^{(-2)}_{F5}$.
We find for example
\bq
\label{eq_ddpsi0}
 \kinvar_1 \frac{\partial^2}{\partial \kinvar_1^2} \Frobeniusbasis
 & = &
 - \Frobeniusbasis
 - \left(2\kinvar_1+1\right) \frac{\partial}{\partial \kinvar_1} \Frobeniusbasis
 - 2 \kinvar_2 \frac{\partial}{\partial \kinvar_2} \Frobeniusbasis
 - 2 \kinvar_3 \frac{\partial}{\partial \kinvar_3} \Frobeniusbasis
 - 2 \kinvar_4 \frac{\partial}{\partial \kinvar_4} \Frobeniusbasis
 - 16 R^{(-2)}_{F5},
 \nonumber \\
 \kinvar_2 \frac{\partial^2}{\partial \kinvar_2^2} \Frobeniusbasis
 & = &
 - \Frobeniusbasis
 - 2 \kinvar_1 \frac{\partial}{\partial \kinvar_1} \Frobeniusbasis
 - \left(2\kinvar_2+1\right)  \frac{\partial}{\partial \kinvar_2} \Frobeniusbasis
 - 2 \kinvar_3 \frac{\partial}{\partial \kinvar_3} \Frobeniusbasis
 - 2 \kinvar_4 \frac{\partial}{\partial \kinvar_4} \Frobeniusbasis
 - 16 R^{(-2)}_{F5},
 \nonumber \\
 \kinvar_3 \frac{\partial^2}{\partial \kinvar_3^2} \Frobeniusbasis
 & = &
 - \Frobeniusbasis
 - 2 \kinvar_1 \frac{\partial}{\partial \kinvar_1} \Frobeniusbasis
 - 2 \kinvar_2 \frac{\partial}{\partial \kinvar_2} \Frobeniusbasis
 - \left(2\kinvar_3+1\right) \frac{\partial}{\partial \kinvar_3} \Frobeniusbasis
 - 2 \kinvar_4 \frac{\partial}{\partial \kinvar_4} \Frobeniusbasis
 - 16 R^{(-2)}_{F5},
 \nonumber \\
 \kinvar_4 \frac{\partial^2}{\partial \kinvar_4^2} \Frobeniusbasis
 & = &
 - \Frobeniusbasis
 - 2 \kinvar_1 \frac{\partial}{\partial \kinvar_1} \Frobeniusbasis
 - 2 \kinvar_2 \frac{\partial}{\partial \kinvar_2} \Frobeniusbasis
 - 2 \kinvar_3 \frac{\partial}{\partial \kinvar_3} \Frobeniusbasis
 - \left(2\kinvar_4+1\right) \frac{\partial}{\partial \kinvar_4} \Frobeniusbasis
 - 16 R^{(-2)}_{F5}.
\eq
The expressions for the mixed derivatives are more lengthy.

From the terms of order $\eps^{-2}$ of the entry $(15,5)$ we obtain the first derivatives of $R^{(-2)}_{F5}$
in terms of $\Frobeniusbasis$,
first derivatives of $\Frobeniusbasis$ and $R^{(-2)}_{F5}$.
The resulting expressions are again rather lengthy.

From eq.~(\ref{eq_ddpsi0}) we may express $R^{(-2)}_{F5}$ in terms of $\Frobeniusbasis$ and its first and second derivatives.
The symmetric expression reads
\bq
\label{def_Rm2_F5}
 R^{(-2)}_{F5}
 & = &
 - \frac{1}{64} \sum_{i=1}^4\left( 
  \kinvar_i \frac{\partial^2}{\partial \kinvar_i^2} + (8y_i + 1) \frac{\partial}{\partial \kinvar_i} +1 \right) \Frobeniusbasis
\eq
We then obtain a set of (linear) differential equations for $\Frobeniusbasis$ alone.
These differential equations generate a differential ideal, which we call the Picard-Fuchs ideal for $\Frobeniusbasis$.
It generalises the concept of a Picard-Fuchs operator from the univariate case to the multivariate case.

Next we turn to the sixteen functions $R^{(-2)}_{66}, \dots, R^{(-2)}_{99}$.
From the terms of order $\eps^0$ of the entries at the positions $(6,6)$, $\dots$, $(9,9)$ 
we determine the first derivatives of these functions in terms of the 
functions $R^{(-2)}_{66}$, $\dots$, $R^{(-2)}_{99}$ and $R^{(-2)}_{F6}, \dots, R^{(-2)}_{F9}$.
We note that the expressions for the first derivatives are linear in
$R^{(-2)}_{66}, \dots, R^{(-2)}_{99}$ and $R^{(-2)}_{F6}, \dots, R^{(-2)}_{F9}$.
From the terms of order $\eps^{-1}$ of the entries at the positions $(F,6), \dots, (F,9)$ we determine the first derivatives of  
$R^{(-2)}_{F6}, \dots, R^{(-2)}_{F9}$.
Again, we find that the derivatives are linear in
$R^{(-2)}_{66}, \dots, R^{(-2)}_{99}$ and $R^{(-2)}_{F6}, \dots, R^{(-2)}_{F9}$.
We may eliminate the functions $R^{(-2)}_{F6}, \dots, R^{(-2)}_{F9}$. 
We find
\bq
\label{def_Rm2_F6_to_F9}
 R^{(-2)}_{Fj}
 & = &
 - \frac{1}{64} \sum\limits_{i=1}^4 \left(
                                \frac{\partial R^{(-2)}_{(5+i)j}}{\partial \kinvar_i} 
                                + \frac{R^{(-2)}_{(5+i)j}}{\Frobeniusbasis} \frac{\partial \Frobeniusbasis}{\partial \kinvar_i}
                                + 8 R^{(-2)}_{(5+i)j} \right), 
 \;\;\;\;\;\;
 j \; \in \; \{6,7,8,9\}.
\eq
We then obtain a set of (linear) differential equations for the sixteen functions $R^{(-2)}_{66}, \dots, R^{(-2)}_{99}$, which we call the Picard-Fuchs ideal
for the set of functions $R^{(-2)}_{66}, \dots, R^{(-2)}_{99}$.
These differential equations may involve $\Frobeniusbasis$.
This is unproblematic, as $\Frobeniusbasis$ can be considered a known function at this stage.

The two remaining unknown functions $R^{(-2)}_{FF}$ and $R^{(-2)}_{FE}$
are determined from first-order differential equations.

\subsection{The transformations $R_2^{(-1)}$ and $R_2^{0}$}

The transformation $R_2^{(-1)}$ removes terms of $B$-order $(-1)$, 
the transformation $R_2^{0}$ removes terms of $B$-order $0$.
The ansatz for $R_2^{(-1)}$ reads
{\footnotesize
\bq
\label{Rm1_ldegree}
 \left(
 \begin{array}{rrrr|r|rrrrrrrr|rr}
 1 & 0 & 0 & 0 & 0 & 0 & 0 & 0 & 0 & 0 & 0 & 0 & 0 & 0 & 0 \\
 0 & 1 & 0 & 0 & 0 & 0 & 0 & 0 & 0 & 0 & 0 & 0 & 0 & 0 & 0 \\
 0 & 0 & 1 & 0 & 0 & 0 & 0 & 0 & 0 & 0 & 0 & 0 & 0 & 0 & 0 \\
 0 & 0 & 0 & 1 & 0 & 0 & 0 & 0 & 0 & 0 & 0 & 0 & 0 & 0 & 0 \\
 \hline
 0 & 0 & 0 & 0 & 1 & 0 & 0 & 0 & 0 & 0 & 0 & 0 & 0 & 0 & 0 \\
 \hline
 0 & 0 & 0 & 0 & R^{(-1)}_{65} & 1 & 0 & 0 & 0 & 0 & 0 & 0 & 0 & 0 & 0 \\
 0 & 0 & 0 & 0 & R^{(-1)}_{75} & 0 & 1 & 0 & 0 & 0 & 0 & 0 & 0 & 0 & 0 \\
 0 & 0 & 0 & 0 & R^{(-1)}_{85} & 0 & 0 & 1 & 0 & 0 & 0 & 0 & 0 & 0 & 0 \\
 0 & 0 & 0 & 0 & R^{(-1)}_{95} & 0 & 0 & 0 & 1 & 0 & 0 & 0 & 0 & 0 & 0 \\
 0 & 0 & 0 & 0 & R^{(-1)}_{A5} & 0 & 0 & 0 & 0 & 1 & 0 & 0 & 0 & 0 & 0 \\
 0 & 0 & 0 & 0 & R^{(-1)}_{B5} & 0 & 0 & 0 & 0 & 0 & 1 & 0 & 0 & 0 & 0 \\
 0 & 0 & 0 & 0 & R^{(-1)}_{C5} & 0 & 0 & 0 & 0 & 0 & 0 & 1 & 0 & 0 & 0 \\
 0 & 0 & 0 & 0 & R^{(-1)}_{D5} & 0 & 0 & 0 & 0 & 0 & 0 & 0 & 1 & 0 & 0 \\
 \hline
 0 & 0 & 0 & 0 & 0 & 0 & 0 & 0 & 0 & 0 & 0 & 0 & 0 & 1 & 0 \\
 0 & 0 & 0 & 0 & \frac{1}{\eps} R^{(-1)}_{F5} & R^{(-1)}_{F6} & R^{(-1)}_{F7}  & R^{(-1)}_{F8}  & R^{(-1)}_{F9} & R^{(-1)}_{FA} & R^{(-1)}_{FB} & R^{(-1)}_{FC} & R^{(-1)}_{FD} & 0 & 1 \\
 \end{array}
 \right).
\eq
}
The ansatz for $R_2^{(0)}$ reads
{\footnotesize
\bq
\label{R0_ldegree}
 \left(
 \begin{array}{rrrr|r|rrrrrrrr|rr}
 1 & 0 & 0 & 0 & 0 & 0 & 0 & 0 & 0 & 0 & 0 & 0 & 0 & 0 & 0 \\
 0 & 1 & 0 & 0 & 0 & 0 & 0 & 0 & 0 & 0 & 0 & 0 & 0 & 0 & 0 \\
 0 & 0 & 1 & 0 & 0 & 0 & 0 & 0 & 0 & 0 & 0 & 0 & 0 & 0 & 0 \\
 0 & 0 & 0 & 1 & 0 & 0 & 0 & 0 & 0 & 0 & 0 & 0 & 0 & 0 & 0 \\
 \hline
 0 & 0 & 0 & 0 & 1 & 0 & 0 & 0 & 0 & 0 & 0 & 0 & 0 & 0 & 0 \\
 \hline
 0 & 0 & 0 & 0 & 0 & 1 & 0 & 0 & 0 & 0 & 0 & 0 & 0 & 0 & 0 \\
 0 & 0 & 0 & 0 & 0 & 0 & 1 & 0 & 0 & 0 & 0 & 0 & 0 & 0 & 0 \\
 0 & 0 & 0 & 0 & 0 & 0 & 0 & 1 & 0 & 0 & 0 & 0 & 0 & 0 & 0 \\
 0 & 0 & 0 & 0 & 0 & 0 & 0 & 0 & 1 & 0 & 0 & 0 & 0 & 0 & 0 \\
 0 & 0 & 0 & 0 & 0 & 0 & 0 & 0 & 0 & 1 & 0 & 0 & 0 & 0 & 0 \\
 0 & 0 & 0 & 0 & 0 & 0 & 0 & 0 & 0 & 0 & 1 & 0 & 0 & 0 & 0 \\
 0 & 0 & 0 & 0 & 0 & 0 & 0 & 0 & 0 & 0 & 0 & 1 & 0 & 0 & 0 \\
 0 & 0 & 0 & 0 & 0 & 0 & 0 & 0 & 0 & 0 & 0 & 0 & 1 & 0 & 0 \\
 \hline
 0 & 0 & 0 & 0 & R^{(0)}_{E5} & 0 & 0 & 0 & 0 & 0 & 0 & 0 & 0 & 1 & 0 \\
 0 & 0 & 0 & 0 & R^{(0)}_{F5} & 0 & 0 & 0 & 0 & 0 & 0 & 0 & 0 & 0 & 1 \\
 \end{array}
 \right).
\eq
}
The unknown functions of the ansatz are determined from first-order differential equations.
From a structural point of view the differential equations for these functions are rather simple: 
We may arrange the order in which we solve them such that at every stage the derivatives of the next unknown function 
are given by expressions involving only known functions at this stage.
Although structurally simple, the explicit expressions are rather lengthy.


\section{The neighbourhood of the point of maximal unipotent monodromy}
\label{sect:mum_point}

We now study the unequal-mass three-loop banana integral in a neighbourhood of the point of maximal unipotent monodromy
$(\kinvar_1,\kinvar_2,\kinvar_3,\kinvar_4)=(0,0,0,0)$.
In a neighbourhood of this point we can give explicit expressions for $\Frobeniusbasis$
and the functions $R^{(-2)}_{66}$, $\dots$, $R^{(-2)}_{99}$.

Let $\bm{n}=(n_1,n_2,n_3,n_4)$ be a multi-index of non-negative integers.
We denote
\bq
 \bm{\kinvar}^{\bm{n}}
 \; = \;
 \kinvar_1^{n_1} \kinvar_2^{n_2} \kinvar_3^{n_3} \kinvar_4^{n_4},
 & &
 \left| \bm{n} \right| \; = \; n_1+n_2+n_3+n_4.
\eq
We set
\bq
\label{def_Frobenius}
 \Frobeniusbasis_{0}
 & = &
 \sum\limits_{n=0}^\infty 
 \sum\limits_{|\bm{n}|=n}
 a_{\bm{n}} \bm{\kinvar}^{\bm{n}},
 \nonumber \\
 \Frobeniusbasis_{1,j}
 & = &
 \frac{1}{\left(2\pi i\right)}
 \sum\limits_{n=0}^\infty 
 \sum\limits_{|\bm{n}|=n}
 \left[
 a_{\bm{n},j} 
 +
 a_{\bm{n}} 
 \ln \kinvar_j
 \right]
 \bm{\kinvar}^{\bm{n}},
 \;\;\;\;\;\;\;\;\;
 1 \; \le \; j \; \le \; 4.
\eq
The coefficients are given by
\bq
\label{Frobenius_coeffs}
 a_{\bm{n}}
 & = &
 \left(-1\right)^{\left|\bm{n}\right|}
 \left( \frac{\left|\bm{n}\right|!}{n_1! n_2! n_3! n_4!} \right)^2,
 \nonumber \\
 a_{\bm{n},j}
 & = &
 \left(-1\right)^{\left|\bm{n}\right|}
 \left( \frac{\left|\bm{n}\right|!}{n_1! n_2! n_3! n_4!} \right)^2
 2 \left[ S_1\left(\left|\bm{n}\right|\right) - S_1\left(n_j\right) \right],
\eq
where $S_m(n)$ denotes the harmonic sum
\bq
 S_m\left(n\right) & = &  \sum\limits_{j=1}^n \frac{1}{j^m}.
\eq
We define the moduli $\tau_1, \tau_2, \tau_3, \tau_4$ by
\bq
 \tau_j & = & 
 \frac{\Frobeniusbasis_{1,j}}{\Frobeniusbasis_{0}}
 \;\;\;\;\;\;\;\;\;
 1 \; \le \; j \; \le \; 4.
\eq
We can furthermore introduce the natural coordinates $\qbar_1, \qbar_2, \qbar_3, \qbar_4$, defined by
\bq
 \qbar_j & = & 
 \exp(2 \pi i \tau_j).
\eq
Expressing the variables $\kinvar_i$ in $\qbar_j$ one obtains for the first few orders~\cite{Candelas:2021lkc}
\begin{equation}
\begin{aligned}
 \kinvar_i(\qbar_j)
 & = \qbar_i\Big( 1
                 + 2(e_1-\qbar_i)
                 +(e_1^2+2e_2-2e_1 \qbar_i+\qbar_i^2) \\
 &\hphantom{{}={}}\hspace{1.5em}
                 +(2e_1e_2+14 e_3-(16e_2-2e_1^2)\qbar_i+10e_1\qbar_i^2-12\qbar_i^3) \\
 &\hphantom{{}={}}\hspace{1.5em}
                 +( e_2^2 + 26 e_1 e_3 - 38 e_4 + ( 2 e_1^3 - 22 e_2 e_1 - 6 e_3 ) \qbar_i + 18 e_1^2 \qbar_i^2 - 20 e_1 \qbar_i^3 )
 \Big) 
 +\mathcal{O}(\mathbf{\qbar}^6). 
\end{aligned}	
\end{equation}
where $e_1=\sum_{i=1}^{4}\qbar_i$, $e_2=\sum_{1\le i<j\le4}\qbar_i \qbar_j$, $e_3=\sum_{1\le i<j<k \le4}\qbar_i \qbar_j \qbar_k$, $e_4 = \qbar_1 \qbar_2 \qbar_3 \qbar_4$ are the elementary symmetric polynomials in the $\qbar_i$.
Note that we have the relation
\bq
 e_4 - e_3 \qbar_i + e_2 \qbar_i^2 - e_1 \qbar_i^3 + \qbar_i^4
 & = & 0,
 \;\;\;\;\;\;
 i \; \in \; \{1,2,3,4\},
\eq
which we may use to eliminate $\qbar_i^4$.
\begin{proposition}
The holomorphic period $\Frobeniusbasis_{0}$ and the single-logarithmic periods $\Frobeniusbasis_{1,j}$
as defined by eq.~(\ref{def_Frobenius}) satisfy the system of
differential equations from section~\ref{sect:step_2} for $\Frobeniusbasis$.
\end{proposition}
\begin{proof}
This is easily verified by inserting the series expansions into the differential equations.
\end{proof}
In the following we will take
\bq
 \Frobeniusbasis & = & \Frobeniusbasis_{0},
\eq
i.e. we normalise the first master integral in the sector of interest by the holomorphic period.
\begin{proposition}
Setting
\bq
 R^{(-2)}_{(5+i)(5+j)} & = & 2 \pi i \; \Frobeniusbasis_{0} \; \kinvar_i \frac{\partial \tau_j}{\partial \kinvar_i}
\eq
satisfies the differential equations from section~\ref{sect:step_2} for $R^{(-2)}_{66}, \dots, R^{(-2)}_{99}$.
\end{proposition}
\begin{proof}
This is again easily verified by inserting the series expansions into the differential equations.
\end{proof}
The functions
\bq
 R^{(-2)}_{65}, R^{(-2)}_{75}, R^{(-2)}_{85}, R^{(-2)}_{95},
 R^{(-2)}_{F5},
 R^{(-2)}_{F6}, R^{(-2)}_{F7}, R^{(-2)}_{F8}, R^{(-2)}_{F9}
\eq
are then given by eqs.~(\ref{def_Rm2_65_to_75},\ref{def_Rm2_F5},\ref{def_Rm2_F6_to_F9}).
All remaining functions are determined by first-order linear differential equations.

We verified that all differential one-forms appearing in the matrix $\tilde{A}$ have at the boundary point at most a simple pole.


\section{Results}
\label{sect:results}

In the basis $K$ the differential equation is in $\eps$-factorised form:
\bq
\label{diff_eq_K}
 d K\left(\eps,\kinvar\right) & = & \eps \tilde{A}(\kinvar) K\left(\eps,\kinvar\right),
\eq
where the non-zero entries of the matrix $\tilde{A}(\kinvar)$ are
{\scriptsize
\bq
\lefteqn{
 \tilde{A}\left(\kinvar\right)
 = } & &
 \nonumber \\
 & &
 \left( \begin{array}{ccccccccccccccc}
 \omega_{1,1} & 0 & 0 & 0 & 0 & 0 & 0 & 0 & 0 & 0 & 0 & 0 & 0 & 0 & 0 \\
 0 & \omega_{2,2} & 0 & 0 & 0 & 0 & 0 & 0 & 0 & 0 & 0 & 0 & 0 & 0 & 0 \\
 0 & 0 & \omega_{3,3} & 0 & 0 & 0 & 0 & 0 & 0 & 0 & 0 & 0 & 0 & 0 & 0 \\
 0 & 0 & 0 & \omega_{4,4} & 0 & 0 & 0 & 0 & 0 & 0 & 0 & 0 & 0 & 0 & 0 \\
 0 & 0 & 0 & 0 & \omega_{5,5} & \omega_{5,6} & \omega_{5,7} & \omega_{5,8} & \omega_{5,9} & 0 & 0 & 0 & 0 & 0 & 0 \\
 0 & 0 & 0 & 0 & \omega_{6,5} & \omega_{6,6} & \omega_{6,7} & \omega_{6,8} & \omega_{6,9} & \omega_{6,10} & \omega_{6,11} & \omega_{6,12} & \omega_{6,13} & \omega_{6,14} & \omega_{6,15} \\
 0 & 0 & 0 & 0 & \omega_{7,5} & \omega_{7,6} & \omega_{7,7} & \omega_{7,8} & \omega_{7,9} & \omega_{7,10} & \omega_{7,11} & \omega_{7,12} & \omega_{7,13} & \omega_{7,14} & \omega_{7,15} \\
 0 & 0 & 0 & 0 & \omega_{8,5} & \omega_{8,6} & \omega_{8,7} & \omega_{8,8} & \omega_{8,9} & \omega_{8,10} & \omega_{8,11} & \omega_{8,12} & \omega_{8,13} & \omega_{8,14} & \omega_{8,15} \\
 0 & 0 & 0 & 0 & \omega_{9,5} & \omega_{9,6} & \omega_{9,7} & \omega_{9,8} & \omega_{9,9} & \omega_{9,10} & \omega_{9,11} & \omega_{9,12} & \omega_{9,13} & \omega_{9,14} & \omega_{9,15} \\
 \omega_{10,1} & \omega_{10,2} & \omega_{10,3} & \omega_{10,4} & \omega_{10,5} & \omega_{10,6} & \omega_{10,7} & \omega_{10,8} & \omega_{10,9} & \omega_{10,10} & \omega_{10,11} & \omega_{10,12} & \omega_{10,13} & 0 & 0 \\
 \omega_{11,1} & \omega_{11,2} & \omega_{11,3} & \omega_{11,4} & \omega_{11,5} & \omega_{11,6} & \omega_{11,7} & \omega_{11,8} & \omega_{11,9} & \omega_{11,10} & \omega_{11,11} & \omega_{11,12} & \omega_{11,13} & 0 & 0 \\
 \omega_{12,1} & \omega_{12,2} & \omega_{12,3} & \omega_{12,4} & \omega_{12,5} & \omega_{12,6} & \omega_{12,7} & \omega_{12,8} & \omega_{12,9} & 0 & \omega_{12,11} & \omega_{12,12} & 0 & \omega_{12,14} & 0 \\
 \omega_{13,1} & \omega_{13,2} & \omega_{13,3} & \omega_{13,4} & \omega_{13,5} & \omega_{13,6} & \omega_{13,7} & \omega_{13,8} & \omega_{13,9} & \omega_{13,10} & \omega_{13,11} & 0 & \omega_{13,13} & \omega_{13,14} & 0 \\
 \omega_{14,1} & 0 & \omega_{14,3} & \omega_{14,4} & \omega_{14,5} & \omega_{14,6} & \omega_{14,7} & \omega_{14,8} & \omega_{14,9} & 0 & 0 & \omega_{14,12} & \omega_{14,13} & \omega_{14,14} & 0 \\
 \omega_{15,1} & \omega_{15,2} & \omega_{15,3} & \omega_{15,4} & \omega_{15,5} & \omega_{15,6} & \omega_{15,7} & \omega_{15,8} & \omega_{15,9} & \omega_{15,10} & \omega_{15,11} & \omega_{15,12} & \omega_{15,13} & \omega_{15,14} & \omega_{15,15} \\
 \end{array} \right).
 \nonumber \\
\eq
}
We will solve the differential equation starting from the boundary point
\bq
 \kinvar_1 \; = \; \kinvar_2 \; = \; \kinvar_3 \; = \; \kinvar_4 \; = \; 0.
\eq
Let $\gamma : [0,1] \rightarrow {\mathbb C}^4$ be a path from the boundary point $(0,0,0,0)$ to $(\kinvar_1,\kinvar_2,\kinvar_3,\kinvar_4)$.
We write
\bq
 f_{i,j}\left(\lambda\right) d\lambda & = & \gamma^\ast \omega_{i,j}
\eq
for the pull-back of $\omega_{i,j}$ by $\gamma$ to the interval $[0,1]$.
We define the $k$-fold iterated integral by
\bq
 \lefteqn{
 I_{\gamma}\left(\omega_{i_1,j_1},\omega_{i_2,j_2},...,\omega_{i_k,j_k};\lambda\right)
 = } & &
 \nonumber \\
 & &
 \lim\limits_{\lambda_0\rightarrow 0}
 R \left[
 \int\limits_{\lambda_0}^{\lambda} d\lambda_1 f_{i_1,j_1}\left(\lambda_1\right)
 \int\limits_{\lambda_0}^{\lambda_1} d\lambda_2 f_{i_2,j_2}\left(\lambda_2\right)
 \dots
 \int\limits_{\lambda_0}^{\lambda_{k-1}} d\lambda_k f_{i_k,j_k}\left(\lambda_r\right)
 \right].
\eq
If $f_{i_k,j_k}(\lambda)$ has a simple pole at $\lambda=0$ 
we employ the standard ``trailing zero'' or ``tangential base point'' regularisation \cite{Brown:2014pnb,Walden:2020odh}:
We first take $\lambda_0$ to have a small non-zero value.
The integration will produce terms with $\ln(\lambda_0)$.
Let $R$ be the operator, which removes all $\ln(\lambda_0)$-terms.
After these terms have been removed, we may take the limit $\lambda_0\rightarrow 0$.
With these preparations we may write the solution to eq.~(\ref{diff_eq_K}) as
\bq
 K\left(\eps,\kinvar\right)
 & = & 
 {\mathcal P}
 \exp\left(\eps \int\limits_\gamma \tilde{A}(\kinvar) \right) C\left(\eps\right).
\eq
where ${\mathcal P}$ is the path ordering operator and $C(\eps)$ the vector of boundary values. 
The boundary values for the tadpoles are
\bq
 C_i\left(\eps\right) & = & e^{3\eps\gamma_E} \; \Gamma\left(1+\eps\right)^3,
 \;\;\;\;\;\;
 i \; \in \; \{1,2,3,4\}.
\eq
By construction, the integrals $J_{10},\dots,J_{14}$ have trivial boundary values. 
As we may choose the functions $R^{(-1)}_{A5}$-$R^{(-1)}_{D5}$ to vanish at our boundary point, this carries over to
the integrals $K_{10},\dots,K_{14}$ and we have
\bq
 C_i\left(\eps\right) & = & 0,
 \;\;\;\;\;\;
 i \; \in \; \{10,11,12,13,14\}.
\eq
More generally, we may choose the boundary constants for the functions from the ansatz such that the matrix $R_2$ reduces to the identity matrix
at the boundary point $(\kinvar_1,\kinvar_2,\kinvar_3,\kinvar_4)=(0,0,0,0)$.
For the integral $J_5$ (and $K_5$) we give the boundary constants and all logarithmic terms:
\bq
\lefteqn{
 \left. J_5 \right|_{\kinvar_1,\kinvar_2,\kinvar_3,\kinvar_4\rightarrow 0}
 \; = \;
 \left. K_5 \right|_{\kinvar_1,\kinvar_2,\kinvar_3,\kinvar_4\rightarrow 0}
 = 
 } \nonumber \\
 & &
 e^{3\eps\gamma_E} \Gamma\left(1+\eps\right)^3 \left[
  \sum\limits_{i<j<k} \kinvar_i^{-\eps} \kinvar_j^{-\eps} \kinvar_k^{-\eps}
  -2 \frac{\Gamma\left(1-\eps\right)^2}{\Gamma\left(1-2\eps\right)} \sum\limits_{i<j} \kinvar_i^{-\eps} \kinvar_j^{-\eps} 
  +3 \frac{\Gamma\left(1-\eps\right)^3\Gamma\left(1+2\eps\right)}{\Gamma\left(1+\eps\right)^2\Gamma\left(1-3\eps\right)} \sum\limits_{i} \kinvar_i^{-\eps} 
 \right. \nonumber \\
 & & \left.
  -4\frac{\Gamma\left(1-\eps\right)^4\Gamma\left(1+3\eps\right)}{\Gamma\left(1+\eps\right)^3\Gamma\left(1-4\eps\right)}
 \right].
\eq
The boundary values for the remaining master integrals $K_6,\dots,K_9$ and $K_{15}$
follow from the higher orders in $\eps$ of this expression.

For the master integral $K_5$ we find
\bq
 K_5
 & = &
 \left[ 16 \zeta_3
 + \sum\limits_{i=6}^{9} \sum\limits_{j=10}^{15} \sum\limits_{k=1}^{4} I_\gamma\left(\omega_{5,i},\omega_{i,j},\omega_{j,k}\right) \right] \eps^3
 + {\mathcal O}\left(\eps^4\right).
\eq
For numerical results we consider the kinematic point
\bq
 \kinvar_1 \; = \; \frac{1}{29},
 \;\;\;\;\;\;
 \kinvar_2 \; = \; \frac{1}{31},
 \;\;\;\;\;\;
 \kinvar_3 \; = \; \frac{1}{37},
 \;\;\;\;\;\;
 \kinvar_4 \; = \; \frac{1}{41},
\eq
which lies within the radius of convergence of the series expansion of $\psi_0$ near the MUM-point.
We obtain at this kinematic point for the integral $K_5$ the value
\bq
 K_5
 & = &
 218.005564 \eps^3  + 983.551161 \eps^4
 +\mathcal{O}\left(\eps^5\right).
\eq
These values have been obtained by expanding all iterated integrals as series up to order $20$ in the integration variable.
We verified numerically against \texttt{AMFlow}~\cite{Liu:2022chg} and sector decomposition programs~\cite{Bogner:2007cr,Heinrich:2023til}
and found perfect agreement for the nine digits given above.

Finally, we consider the phenomenologically interesting case, where the four unequal masses are given by the four heavy particles of the standard
model:
\bq
 & & m_1 \; = \; m_W \; = \; 80.37 \; \mbox{GeV},
 \;\;\; 
 m_2 \; = \; m_Z \; = \; 91.19 \; \mbox{GeV},
 \nonumber \\
 & & m_3 \; = \; m_H \; = \; 125.20 \; \mbox{GeV},
 \;\;\; 
 m_4 \; = \; m_t \; = \; 172.56 \; \mbox{GeV}.
\eq
In fig.~\ref{fig:plot_K_5} we plot the real and the imaginary part of the $\eps^3$-term and the $\eps^4$-term of the master integral $K_5$
as a function of
\bq
 \lambda 
 & = &
 \frac{\left(m_W+m_Z+m_H+m_t\right)^2}{\left(-p^2\right)}.
\eq
\begin{figure}
\begin{center}
\includegraphics[scale=0.62]{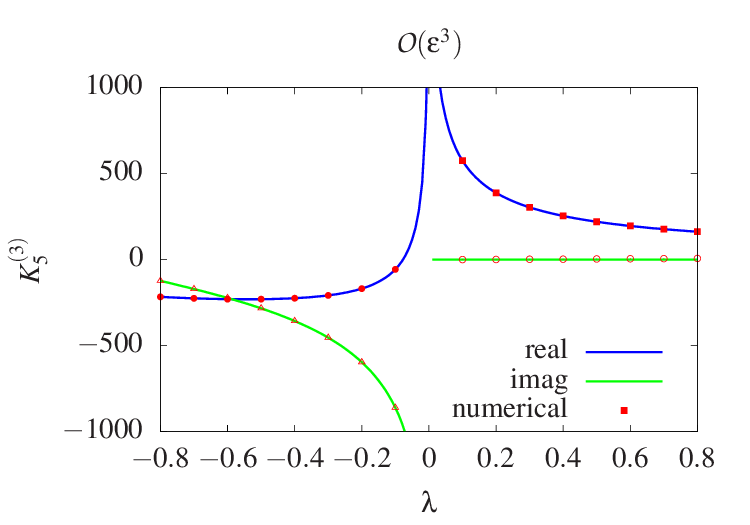}
\includegraphics[scale=0.62]{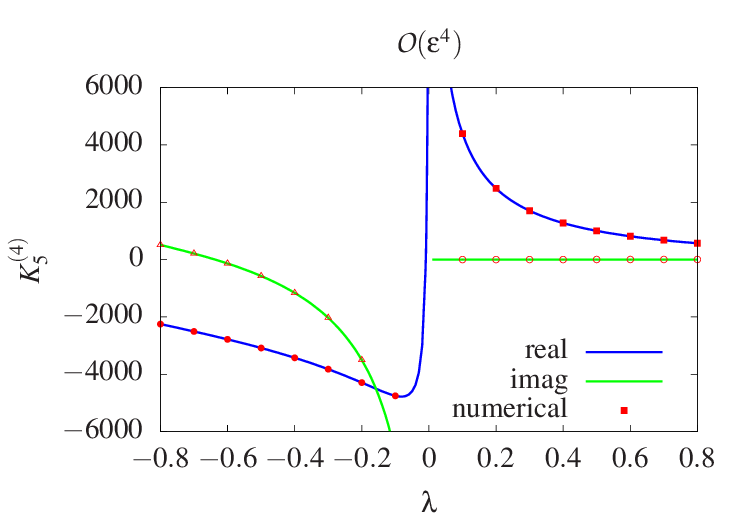}
\end{center}
\caption{
The real and the imaginary part of the master integral $K_5$ at order $\eps^3$ (left plot) and order $\eps^4$ (right plot) as a function of $\lambda$.
}
\label{fig:plot_K_5}
\end{figure}
The plot also shows numerical results from \texttt{AMFlow}~\cite{Liu:2022chg} and sector decomposition programs~\cite{Bogner:2007cr,Heinrich:2023til},
which are in perfect agreement.


\section{Conclusions}
\label{sect:conclusions}

In this paper we considered the unequal-mass three-loop banana integral.
For several reasons, this is a challenging integral: It depends on four kinematic variables, 
it has a sector with a large number of master integrals (eleven master integrals in the top sector) and it is associated to a K3-geometry.
We showed how this integral can be computed systematically with the general algorithm of ref.~\cite{Bree:2025maw}.
The calculation in this paper illustrates how the algorithm works for a highly non-trivial example.
It would be interesting to extend the present analysis to banana integrals with distinct masses at any loop order.

\subsection*{Acknowledgements}

This work has been supported by the Research Unit ``Modern Foundations of Scattering Amplitudes'' (FOR 5582)
funded by the German Research Foundation (DFG).
X.W. is supported by the University Development Fund of The Chinese University of Hong Kong, Shenzhen, under the Grant No. UDF01003912.


\begin{appendix}

\section{The integrals $J_{10},\dots,J_{14}$}
\label{appendix:J_10_J_14}

In this appendix we give the explicit expressions for $J_{10},\dots,J_{14}$ in terms of the basis $I$.
The basis $I$ is defined in eq.~(\ref{def_basis_I}).
For $J_{10},\dots,J_{14}$ we write
\bq
 J_i & = & \sum\limits_{j=1}^{15} \frac{c_{i,j}}{3\left(1+\kinvar_1+\kinvar_2+\kinvar_3+\kinvar_4\right)} I_j.
\eq
Only the coefficients with $j \ge 5$ are non-vanishing.
The coefficients $c_{i,j}$ are given by
{\footnotesize
\bq
c_{10,5} & = &  \left( 45\,{\kinvar_{1}}^{2}-102\,\kinvar_{1}\,\kinvar_{2}+204\,\kinvar_{1}\,\kinvar_{3}+204\,\kinvar_{1}\,\kinvar_{4}-27\,{\kinvar_{2}}^{2}-60\,
\kinvar_{2}\,\kinvar_{3}-60\,\kinvar_{2}\,\kinvar_{4}-9\,{\kinvar_{3}}^{2}-186\,\kinvar_{3}\,\kinvar_{4}-9\,{\kinvar_{4}}^{2}
 \right. \nonumber \\
 & & \left.
 +45\,\kinvar_{1}-27\,\kinvar_{2}-9\,\kinvar_{3}-9\,
\kinvar_{4} \right) {\eps}^{3}+ \left( 9\,{\kinvar_{1}}^{2}-68\,\kinvar_{1}\,\kinvar_{2}+106\,\kinvar_{1}\,\kinvar_{3}+106\,\kinvar_{1}\,\kinvar_{4}-7\,{\kinvar_{2}
}^{2}-22\,\kinvar_{2}\,\kinvar_{3}-22\,\kinvar_{2}\,\kinvar_{4}
 \right. \nonumber \\
 & & \left.
 -{\kinvar_{3}}^{2}-100\,\kinvar_{3}\,\kinvar_{4}-{\kinvar_{4}}^{2}+9\,\kinvar_{1}-7\,\kinvar_{2}-\kinvar_{3}-
\kinvar_{4} \right) {\eps}^{2}+ \left( -10\,\kinvar_{1}\,\kinvar_{2}+14\,\kinvar_{1}\,\kinvar_{3}+14\,\kinvar_{1}\,\kinvar_{4}-2\,\kinvar_{2}\,\kinvar_{3}-2\,\kinvar_{2}
\,\kinvar_{4}
 \right. \nonumber \\
 & & \left.
-14\,\kinvar_{3}\,\kinvar_{4} \right) \eps,
 \nonumber \\
c_{10,6} & = & -3\, \left( 3\,{\kinvar_{1}}^{2}-12\,\kinvar_{1}\,\kinvar_{2}+24\,\kinvar_{1}\,\kinvar_{3}+24\,\kinvar_{1}\,\kinvar_{4}+5\,{\kinvar_{2}}^{2}-6\,\kinvar_{2}
\,\kinvar_{3}-6\,\kinvar_{2}\,\kinvar_{4}-7\,{\kinvar_{3}}^{2}-42\,\kinvar_{3}\,\kinvar_{4}-7\,{\kinvar_{4}}^{2}+6\,\kinvar_{1}+8\,\kinvar_{2}
 \right. \nonumber \\
 & & \left.
-4\,\kinvar_{3}-4\,\kinvar_{4}+3
 \right) {\eps}^{2}\kinvar_{1}+ \left( 15\,\kinvar_{1}\,\kinvar_{2}-21\,\kinvar_{1}\,\kinvar_{3}-21\,\kinvar_{1}\,\kinvar_{4}-5\,{\kinvar_{2}}^{2}+6\,\kinvar_{2}\,
\kinvar_{3}+6\,\kinvar_{2}\,\kinvar_{4}+7\,{\kinvar_{3}}^{2}+42\,\kinvar_{3}\,\kinvar_{4}
 \right. \nonumber \\
 & & \left.
+7\,{\kinvar_{4}}^{2}
-5\,\kinvar_{2}+7\,\kinvar_{3}+7\,\kinvar_{4} \right) \eps\,
\kinvar_{1},
 \nonumber \\
c_{10,7} & = & - \left( 31\,{\kinvar_{1}}^{2}-36\,\kinvar_{1}\,\kinvar_{2}+110\,\kinvar_{1}\,\kinvar_{3}+110\,\kinvar_{1}\,\kinvar_{4}-7\,{\kinvar_{2}}^{2}-24\,
\kinvar_{2}\,\kinvar_{3}-24\,\kinvar_{2}\,\kinvar_{4}-5\,{\kinvar_{3}}^{2}-94\,\kinvar_{3}\,\kinvar_{4}-5\,{\kinvar_{4}}^{2}+24\,\kinvar_{1}
 \right. \nonumber \\
 & & \left.
-14\,\kinvar_{2}-12\,\kinvar_{3}-12
\,\kinvar_{4}-7 \right) {\eps}^{2}\kinvar_{2}- \left( 5\,{\kinvar_{1}}^{2}-15\,\kinvar_{1}\,\kinvar_{2}+34\,\kinvar_{1}\,\kinvar_{3}+34\,\kinvar_{1}\,\kinvar_{4}-3
\,\kinvar_{2}\,\kinvar_{3}-3\,\kinvar_{2}\,\kinvar_{4}+{\kinvar_{3}}^{2}
 \right. \nonumber \\
 & & \left.
-26\,\kinvar_{3}\,\kinvar_{4}+{\kinvar_{4}}^{2}+5\,\kinvar_{1}+\kinvar_{3}+\kinvar_{4} \right) \eps\,
\kinvar_{2},
 \nonumber \\
c_{10,8} & = &  \left( 11\,{\kinvar_{1}}^{2}+82\,\kinvar_{1}\,\kinvar_{2}-72\,\kinvar_{1}\,\kinvar_{3}-98\,\kinvar_{1}\,\kinvar_{4}+11\,{\kinvar_{2}}^{2}+24\,\kinvar_{2}
\,\kinvar_{3}-2\,\kinvar_{2}\,\kinvar_{4}+{\kinvar_{3}}^{2}+60\,\kinvar_{3}\,\kinvar_{4}-25\,{\kinvar_{4}}^{2}+12\,\kinvar_{1}
 \right. \nonumber \\
 & & \left.
+12\,\kinvar_{2}+2\,\kinvar_{3}-24\,\kinvar_{4}+1
 \right) {\eps}^{2}\kinvar_{3}+ \left( 7\,{\kinvar_{1}}^{2}+26\,\kinvar_{1}\,\kinvar_{2}-21\,\kinvar_{1}\,\kinvar_{3}-28\,\kinvar_{1}\,\kinvar_{4}-{\kinvar_{2}}^{2}
+3\,\kinvar_{2}\,\kinvar_{3}-4\,\kinvar_{2}\,\kinvar_{4}
 \right. \nonumber \\
 & & \left.
+21\,\kinvar_{3}\,\kinvar_{4}
-7\,{\kinvar_{4}}^{2}+7\,\kinvar_{1}-\kinvar_{2}-7\,\kinvar_{4} \right) \eps\,\kinvar_{3},
 \nonumber \\
c_{10,9} & = &  \left( 11\,{\kinvar_{1}}^{2}+82\,\kinvar_{1}\,\kinvar_{2}-98\,\kinvar_{1}\,\kinvar_{3}-72\,\kinvar_{1}\,\kinvar_{4}+11\,{\kinvar_{2}}^{2}-2\,\kinvar_{2}\,
\kinvar_{3}+24\,\kinvar_{2}\,\kinvar_{4}-25\,{\kinvar_{3}}^{2}+60\,\kinvar_{3}\,\kinvar_{4}+{\kinvar_{4}}^{2}+12\,\kinvar_{1}
 \right. \nonumber \\
 & & \left.
+12\,\kinvar_{2}-24\,\kinvar_{3}+2\,\kinvar_{4}+1
 \right) {\eps}^{2}\kinvar_{4}+ \left( 7\,{\kinvar_{1}}^{2}+26\,\kinvar_{1}\,\kinvar_{2}-28\,\kinvar_{1}\,\kinvar_{3}-21\,\kinvar_{1}\,\kinvar_{4}-{\kinvar_{2}}^{2}
-4\,\kinvar_{2}\,\kinvar_{3}+3\,\kinvar_{2}\,\kinvar_{4}-7\,{\kinvar_{3}}^{2}
 \right. \nonumber \\
 & & \left.
+21\,\kinvar_{3}\,\kinvar_{4}+7\,\kinvar_{1}-\kinvar_{2}-7\,\kinvar_{3} \right) \eps\,\kinvar_{4},
 \nonumber \\
c_{10,10} & = &  \left( 5\,{\kinvar_{1}}^{2}-10\,\kinvar_{1}\,\kinvar_{2}+26\,\kinvar_{1}\,\kinvar_{3}+26\,\kinvar_{1}\,\kinvar_{4}+5\,{\kinvar_{2}}^{2}-6\,\kinvar_{2}\,
\kinvar_{3}-6\,\kinvar_{2}\,\kinvar_{4}-7\,{\kinvar_{3}}^{2}-42\,\kinvar_{3}\,\kinvar_{4}-7\,{\kinvar_{4}}^{2}+10\,\kinvar_{1}+10\,\kinvar_{2}
 \right. \nonumber \\
 & & \left.
-2\,\kinvar_{3}-2\,\kinvar_{4}+5
 \right) \eps\,\kinvar_{1}\,\kinvar_{2},
 \nonumber \\
c_{10,11} & = & - \left( 7\,{\kinvar_{1}}^{2}+22\,\kinvar_{1}\,\kinvar_{2}-14\,\kinvar_{1}\,\kinvar_{3}-14\,\kinvar_{1}\,\kinvar_{4}-5\,{\kinvar_{2}}^{2}+6\,\kinvar_{2}\,
\kinvar_{3}+6\,\kinvar_{2}\,\kinvar_{4}+7\,{\kinvar_{3}}^{2}+42\,\kinvar_{3}\,\kinvar_{4}+7\,{\kinvar_{4}}^{2}+14\,\kinvar_{1}
 \right. \nonumber \\
 & & \left.
+2\,\kinvar_{2}
+14\,\kinvar_{3}+14\,\kinvar_{4}+7
 \right) \eps\,\kinvar_{1}\,\kinvar_{3},
 \nonumber \\
c_{10,12} & = & - \left( 7\,{\kinvar_{1}}^{2}+22\,\kinvar_{1}\,\kinvar_{2}-14\,\kinvar_{1}\,\kinvar_{3}-14\,\kinvar_{1}\,\kinvar_{4}-5\,{\kinvar_{2}}^{2}+6\,\kinvar_{2}\,
\kinvar_{3}+6\,\kinvar_{2}\,\kinvar_{4}+7\,{\kinvar_{3}}^{2}+42\,\kinvar_{3}\,\kinvar_{4}+7\,{\kinvar_{4}}^{2}+14\,\kinvar_{1}
 \right. \nonumber \\
 & & \left.
+2\,\kinvar_{2}
+14\,\kinvar_{3}+14\,\kinvar_{4}+7
 \right) \eps\,\kinvar_{1}\,\kinvar_{4},
 \nonumber \\
c_{10,13} & = & - \left( 3\,{\kinvar_{1}}^{2}+22\,\kinvar_{1}\,\kinvar_{2}-26\,\kinvar_{1}\,\kinvar_{3}-42\,\kinvar_{1}\,\kinvar_{4}-{\kinvar_{2}}^{2}+2\,\kinvar_{2}\,
\kinvar_{3}-14\,\kinvar_{2}\,\kinvar_{4}-{\kinvar_{3}}^{2}+10\,\kinvar_{3}\,\kinvar_{4}-17\,{\kinvar_{4}}^{2}+2\,\kinvar_{1}-2\,\kinvar_{2}
 \right. \nonumber \\
 & & \left.
-2\,\kinvar_{3}-18\,\kinvar_{4}-1
 \right) \eps\,\kinvar_{2}\,\kinvar_{3},
 \nonumber \\
c_{10,14} & = & - \left( 3\,{\kinvar_{1}}^{2}+22\,\kinvar_{1}\,\kinvar_{2}-42\,\kinvar_{1}\,\kinvar_{3}-26\,\kinvar_{1}\,\kinvar_{4}-{\kinvar_{2}}^{2}-14\,\kinvar_{2}\,
\kinvar_{3}+2\,\kinvar_{2}\,\kinvar_{4}-17\,{\kinvar_{3}}^{2}+10\,\kinvar_{3}\,\kinvar_{4}-{\kinvar_{4}}^{2}+2\,\kinvar_{1}-2\,\kinvar_{2}
 \right. \nonumber \\
 & & \left.
-18\,\kinvar_{3}-2\,\kinvar_{4}-1
 \right) \eps\,\kinvar_{2}\,\kinvar_{4},
 \nonumber \\
c_{10,15} & = & - \left( 21\,{\kinvar_{1}}^{2}+30\,\kinvar_{1}\,\kinvar_{2}-14\,\kinvar_{1}\,\kinvar_{3}-14\,\kinvar_{1}\,\kinvar_{4}-11\,{\kinvar_{2}}^{2}-14\,
\kinvar_{2}\,\kinvar_{3}-14\,\kinvar_{2}\,\kinvar_{4}-7\,{\kinvar_{3}}^{2}+14\,\kinvar_{3}\,\kinvar_{4}-7\,{\kinvar_{4}}^{2}+14\,\kinvar_{1}
 \right. \nonumber \\
 & & \left.
-18\,\kinvar_{2}-14\,\kinvar_{3}-14
\,\kinvar_{4}-7 \right) \eps\,\kinvar_{3}\,\kinvar_{4},
 \nonumber \\
c_{11,5} & = &  \left( -96\,\kinvar_{1}\,\kinvar_{2}+48\,\kinvar_{1}\,\kinvar_{3}+48\,\kinvar_{1}\,\kinvar_{4}+24\,\kinvar_{2}\,\kinvar_{3}+24\,\kinvar_{2}\,\kinvar_{4}-48\,
\kinvar_{3}\,\kinvar_{4} \right) {\eps}^{3}+ \left( -58\,\kinvar_{1}\,\kinvar_{2}+29\,\kinvar_{1}\,\kinvar_{3}+29\,\kinvar_{1}\,\kinvar_{4}
 \right. \nonumber \\
 & & \left.
-2\,{\kinvar_{2}}^{2}+13
\,\kinvar_{2}\,\kinvar_{3}+13\,\kinvar_{2}\,\kinvar_{4}+{\kinvar_{3}}^{2}-26\,\kinvar_{3}\,\kinvar_{4}+{\kinvar_{4}}^{2}-2\,\kinvar_{2}+\kinvar_{3}+\kinvar_{4} \right) {\eps}^{
2}+ \left( -8\,\kinvar_{1}\,\kinvar_{2}+4\,\kinvar_{1}\,\kinvar_{3}+4\,\kinvar_{1}\,\kinvar_{4}
 \right. \nonumber \\
 & & \left.
+2\,\kinvar_{2}\,\kinvar_{3}+2\,\kinvar_{2}\,\kinvar_{4}-4\,\kinvar_{3}\,\kinvar_{4}
 \right) \eps,
 \nonumber \\
c_{11,6} & = & 6\, \left( 6\,\kinvar_{1}\,\kinvar_{2}-3\,\kinvar_{1}\,\kinvar_{3}-3\,\kinvar_{1}\,\kinvar_{4}-2\,{\kinvar_{2}}^{2}-3\,\kinvar_{2}\,\kinvar_{3}-3\,\kinvar_{2}\,
\kinvar_{4}+{\kinvar_{3}}^{2}+6\,\kinvar_{3}\,\kinvar_{4}+{\kinvar_{4}}^{2}-2\,\kinvar_{2}+\kinvar_{3}+\kinvar_{4} \right) {\eps}^{2}\kinvar_{1}
 \nonumber \\
 & & 
+2\, \left( 6\,\kinvar_{1}
\,\kinvar_{2}-3\,\kinvar_{1}\,\kinvar_{3}-3\,\kinvar_{1}\,\kinvar_{4}-2\,{\kinvar_{2}}^{2}-3\,\kinvar_{2}\,\kinvar_{3}-3\,\kinvar_{2}\,\kinvar_{4}+{\kinvar_{3}}^{2}+6\,\kinvar_{3}\,
\kinvar_{4}+{\kinvar_{4}}^{2}-2\,\kinvar_{2}+\kinvar_{3}+\kinvar_{4} \right) \eps\,\kinvar_{1},
 \nonumber \\
c_{11,7} & = & -2\, \left( 7\,{\kinvar_{1}}^{2}-18\,\kinvar_{1}\,\kinvar_{2}+17\,\kinvar_{1}\,\kinvar_{3}+17\,\kinvar_{1}\,\kinvar_{4}-{\kinvar_{2}}^{2}+3\,\kinvar_{2}\,
\kinvar_{3}+3\,\kinvar_{2}\,\kinvar_{4}-2\,{\kinvar_{3}}^{2}-16\,\kinvar_{3}\,\kinvar_{4}-2\,{\kinvar_{4}}^{2}+6\,\kinvar_{1}-2\,\kinvar_{2}
 \right. \nonumber \\
 & & \left.
-3\,\kinvar_{3}-3\,\kinvar_{4}-1
 \right) {\eps}^{2}\kinvar_{2}- \left( 4\,{\kinvar_{1}}^{2}-12\,\kinvar_{1}\,\kinvar_{2}+11\,\kinvar_{1}\,\kinvar_{3}+11\,\kinvar_{1}\,\kinvar_{4}+3\,\kinvar_{2}\,
\kinvar_{3}+3\,\kinvar_{2}\,\kinvar_{4}-{\kinvar_{3}}^{2}-10\,\kinvar_{3}\,\kinvar_{4}-{\kinvar_{4}}^{2}
 \right. \nonumber \\
 & & \left.
+4\,\kinvar_{1}-\kinvar_{3}-\kinvar_{4} \right) \eps\,\kinvar_{2},
 \nonumber \\
c_{11,8} & = &  \left( 7\,{\kinvar_{1}}^{2}+62\,\kinvar_{1}\,\kinvar_{2}-18\,\kinvar_{1}\,\kinvar_{3}-28\,\kinvar_{1}\,\kinvar_{4}+7\,{\kinvar_{2}}^{2}-6\,\kinvar_{2}\,
\kinvar_{3}-16\,\kinvar_{2}\,\kinvar_{4}-{\kinvar_{3}}^{2}+12\,\kinvar_{3}\,\kinvar_{4}-11\,{\kinvar_{4}}^{2}+6\,\kinvar_{1}+6\,\kinvar_{2}
 \right. \nonumber \\
 & & \left.
-2\,\kinvar_{3}-12\,\kinvar_{4}-1
 \right) {\eps}^{2}\kinvar_{3}+ \left( 2\,{\kinvar_{1}}^{2}+19\,\kinvar_{1}\,\kinvar_{2}-6\,\kinvar_{1}\,\kinvar_{3}-8\,\kinvar_{1}\,\kinvar_{4}+{\kinvar_{2}}^{2}-3
\,\kinvar_{2}\,\kinvar_{3}-5\,\kinvar_{2}\,\kinvar_{4}+6\,\kinvar_{3}\,\kinvar_{4}-2\,{\kinvar_{4}}^{2}
 \right. \nonumber \\
 & & \left.
+2\,\kinvar_{1}+\kinvar_{2}-2\,\kinvar_{4} \right) \eps\,\kinvar_{3},
 \nonumber \\
c_{11,9} & = &  \left( 7\,{\kinvar_{1}}^{2}+62\,\kinvar_{1}\,\kinvar_{2}-28\,\kinvar_{1}\,\kinvar_{3}-18\,\kinvar_{1}\,\kinvar_{4}+7\,{\kinvar_{2}}^{2}-16\,\kinvar_{2}\,
\kinvar_{3}-6\,\kinvar_{2}\,\kinvar_{4}-11\,{\kinvar_{3}}^{2}+12\,\kinvar_{3}\,\kinvar_{4}-{\kinvar_{4}}^{2}+6\,\kinvar_{1}+6\,\kinvar_{2}
 \right. \nonumber \\
 & & \left.
-12\,\kinvar_{3}-2\,\kinvar_{4}-1
 \right) {\eps}^{2}\kinvar_{4}+ \left( 2\,{\kinvar_{1}}^{2}+19\,\kinvar_{1}\,\kinvar_{2}-8\,\kinvar_{1}\,\kinvar_{3}-6\,\kinvar_{1}\,\kinvar_{4}+{\kinvar_{2}}^{2}-5
\,\kinvar_{2}\,\kinvar_{3}-3\,\kinvar_{2}\,\kinvar_{4}-2\,{\kinvar_{3}}^{2}+6\,\kinvar_{3}\,\kinvar_{4}
 \right. \nonumber \\
 & & \left.
+2\,\kinvar_{1}+\kinvar_{2}-2\,\kinvar_{3} \right) \eps\,\kinvar_{4},
 \nonumber \\
c_{11,10} & = & 2\, \left( 2\,{\kinvar_{1}}^{2}-4\,\kinvar_{1}\,\kinvar_{2}+5\,\kinvar_{1}\,\kinvar_{3}+5\,\kinvar_{1}\,\kinvar_{4}+2\,{\kinvar_{2}}^{2}+3\,\kinvar_{2}\,
\kinvar_{3}+3\,\kinvar_{2}\,\kinvar_{4}-{\kinvar_{3}}^{2}-6\,\kinvar_{3}\,\kinvar_{4}-{\kinvar_{4}}^{2}+4\,\kinvar_{1}+4\,\kinvar_{2}+\kinvar_{3}+\kinvar_{4}
 \right. \nonumber \\
 & & \left.
+2 \right) \eps\,
\kinvar_{1}\,\kinvar_{2},
 \nonumber \\
c_{11,11} & = & -2\, \left( {\kinvar_{1}}^{2}+7\,\kinvar_{1}\,\kinvar_{2}-2\,\kinvar_{1}\,\kinvar_{3}-2\,\kinvar_{1}\,\kinvar_{4}-2\,{\kinvar_{2}}^{2}-3\,\kinvar_{2}\,
\kinvar_{3}-3\,\kinvar_{2}\,\kinvar_{4}+{\kinvar_{3}}^{2}+6\,\kinvar_{3}\,\kinvar_{4}+{\kinvar_{4}}^{2}+2\,\kinvar_{1}-\kinvar_{2}+2\,\kinvar_{3}
 \right. \nonumber \\
 & & \left.
+2\,\kinvar_{4}+1 \right) \eps
\,\kinvar_{1}\,\kinvar_{3},
 \nonumber \\
c_{11,12} & = & -2\, \left( {\kinvar_{1}}^{2}+7\,\kinvar_{1}\,\kinvar_{2}-2\,\kinvar_{1}\,\kinvar_{3}-2\,\kinvar_{1}\,\kinvar_{4}-2\,{\kinvar_{2}}^{2}-3\,\kinvar_{2}\,
\kinvar_{3}-3\,\kinvar_{2}\,\kinvar_{4}+{\kinvar_{3}}^{2}+6\,\kinvar_{3}\,\kinvar_{4}+{\kinvar_{4}}^{2}+2\,\kinvar_{1}-\kinvar_{2}+2\,\kinvar_{3}
 \right. \nonumber \\
 & & \left.
+2\,\kinvar_{4}+1 \right) \eps
\,\kinvar_{1}\,\kinvar_{4},
 \nonumber \\
c_{11,13} & = &  \left( 3\,{\kinvar_{1}}^{2}-14\,\kinvar_{1}\,\kinvar_{2}+10\,\kinvar_{1}\,\kinvar_{3}+12\,\kinvar_{1}\,\kinvar_{4}-{\kinvar_{2}}^{2}+2\,\kinvar_{2}\,
\kinvar_{3}+4\,\kinvar_{2}\,\kinvar_{4}-{\kinvar_{3}}^{2}-8\,\kinvar_{3}\,\kinvar_{4}+{\kinvar_{4}}^{2}+2\,\kinvar_{1}-2\,\kinvar_{2}-2\,\kinvar_{3}
 \right. \nonumber \\
 & & \left.
-1 \right) \eps\,\kinvar_{2}
\,\kinvar_{3},
 \nonumber \\
c_{11,14} & = &  \left( 3\,{\kinvar_{1}}^{2}-14\,\kinvar_{1}\,\kinvar_{2}+12\,\kinvar_{1}\,\kinvar_{3}+10\,\kinvar_{1}\,\kinvar_{4}-{\kinvar_{2}}^{2}+4\,\kinvar_{2}\,
\kinvar_{3}+2\,\kinvar_{2}\,\kinvar_{4}+{\kinvar_{3}}^{2}-8\,\kinvar_{3}\,\kinvar_{4}-{\kinvar_{4}}^{2}+2\,\kinvar_{1}-2\,\kinvar_{2}-2\,\kinvar_{4}
 \right. \nonumber \\
 & & \left.
-1 \right) \eps\,\kinvar_{2}
\,\kinvar_{4},
 \nonumber \\
c_{11,15} & = & -2\, \left( 3\,{\kinvar_{1}}^{2}+12\,\kinvar_{1}\,\kinvar_{2}-2\,\kinvar_{1}\,\kinvar_{3}-2\,\kinvar_{1}\,\kinvar_{4}+{\kinvar_{2}}^{2}-2\,\kinvar_{2}\,
\kinvar_{3}-2\,\kinvar_{2}\,\kinvar_{4}-{\kinvar_{3}}^{2}+2\,\kinvar_{3}\,\kinvar_{4}-{\kinvar_{4}}^{2}+2\,\kinvar_{1}-2\,\kinvar_{3}-2\,\kinvar_{4}
 \right. \nonumber \\
 & & \left.
-1 \right) \eps\,\kinvar_{3}
\,\kinvar_{4},
 \nonumber \\
c_{12,5} & = &  \left( -24\,\kinvar_{1}\,\kinvar_{2}+12\,\kinvar_{1}\,\kinvar_{3}+12\,\kinvar_{1}\,\kinvar_{4}-24\,{\kinvar_{2}}^{2}-84\,\kinvar_{2}\,\kinvar_{3}-84\,
\kinvar_{2}\,\kinvar_{4}+12\,{\kinvar_{3}}^{2}+168\,\kinvar_{3}\,\kinvar_{4}+12\,{\kinvar_{4}}^{2}-24\,\kinvar_{2}
 \right. \nonumber \\
 & & \left.
+12\,\kinvar_{3}
+12\,\kinvar_{4} \right) {\eps}^{3
}+ \left( -6\,\kinvar_{1}\,\kinvar_{2}+3\,\kinvar_{1}\,\kinvar_{3}+3\,\kinvar_{1}\,\kinvar_{4}-6\,{\kinvar_{2}}^{2}-45\,\kinvar_{2}\,\kinvar_{3}-45\,\kinvar_{2}\,\kinvar_{4}+3
\,{\kinvar_{3}}^{2}+90\,\kinvar_{3}\,\kinvar_{4}+3\,{\kinvar_{4}}^{2}
 \right. \nonumber \\
 & & \left.
-6\,\kinvar_{2}
+3\,\kinvar_{3}+3\,\kinvar_{4} \right) {\eps}^{2}+ \left( -6\,\kinvar_{2}
\,\kinvar_{3}-6\,\kinvar_{2}\,\kinvar_{4}+12\,\kinvar_{3}\,\kinvar_{4} \right) \eps,
 \nonumber \\
c_{12,6} & = & 6\, \left( 2\,\kinvar_{1}\,\kinvar_{2}-\kinvar_{1}\,\kinvar_{3}-\kinvar_{1}\,\kinvar_{4}+2\,{\kinvar_{2}}^{2}+7\,\kinvar_{2}\,\kinvar_{3}+7\,\kinvar_{2}\,\kinvar_{4}-{
\kinvar_{3}}^{2}-14\,\kinvar_{3}\,\kinvar_{4}-{\kinvar_{4}}^{2}+2\,\kinvar_{2}-\kinvar_{3}-\kinvar_{4} \right) {\eps}^{2}\kinvar_{1}
 \nonumber \\
 & & 
+12\, \left( \kinvar_{2}\,\kinvar_{3}+
\kinvar_{2}\,\kinvar_{4}-2\,\kinvar_{3}\,\kinvar_{4} \right) \eps\,\kinvar_{1},
 \nonumber \\
c_{12,7} & = & 6\, \left( {\kinvar_{1}}^{2}+2\,\kinvar_{1}\,\kinvar_{2}-\kinvar_{1}\,\kinvar_{3}-\kinvar_{1}\,\kinvar_{4}+{\kinvar_{2}}^{2}+5\,\kinvar_{2}\,\kinvar_{3}+5\,\kinvar_{2}
\,\kinvar_{4}-2\,{\kinvar_{3}}^{2}-16\,\kinvar_{3}\,\kinvar_{4}-2\,{\kinvar_{4}}^{2}+2\,\kinvar_{1}+2\,\kinvar_{2}-\kinvar_{3}-\kinvar_{4}
 \right. \nonumber \\
 & & \left.
+1 \right) {\eps}^{2}\kinvar_{2}
-3\, \left( \kinvar_{1}\,\kinvar_{3}+\kinvar_{1}\,\kinvar_{4}-3\,\kinvar_{2}\,\kinvar_{3}-3\,\kinvar_{2}\,\kinvar_{4}+{\kinvar_{3}}^{2}+10\,\kinvar_{3}\,\kinvar_{4}+{\kinvar_{4}}^{2}
+\kinvar_{3}+\kinvar_{4} \right) \eps\,\kinvar_{2},
 \nonumber \\
c_{12,8} & = & -3\, \left( {\kinvar_{1}}^{2}+2\,\kinvar_{1}\,\kinvar_{2}+2\,\kinvar_{1}\,\kinvar_{3}-4\,\kinvar_{1}\,\kinvar_{4}+{\kinvar_{2}}^{2}-10\,\kinvar_{2}\,\kinvar_{3}-
16\,\kinvar_{2}\,\kinvar_{4}+{\kinvar_{3}}^{2}+20\,\kinvar_{3}\,\kinvar_{4}-5\,{\kinvar_{4}}^{2}+2\,\kinvar_{1}+2\,\kinvar_{2}
 \right. \nonumber \\
 & & \left.
+2\,\kinvar_{3}
-4\,\kinvar_{4}+1 \right) {
\eps}^{2}\kinvar_{3}-3\, \left( \kinvar_{1}\,\kinvar_{2}-2\,\kinvar_{1}\,\kinvar_{4}+{\kinvar_{2}}^{2}-3\,\kinvar_{2}\,\kinvar_{3}-5\,\kinvar_{2}\,\kinvar_{4}+6\,\kinvar_{3}\,
\kinvar_{4}-2\,{\kinvar_{4}}^{2}+\kinvar_{2}-2\,\kinvar_{4} \right) \eps\,\kinvar_{3},
 \nonumber \\
c_{12,9} & = & -3\, \left( {\kinvar_{1}}^{2}+2\,\kinvar_{1}\,\kinvar_{2}-4\,\kinvar_{1}\,\kinvar_{3}+2\,\kinvar_{1}\,\kinvar_{4}+{\kinvar_{2}}^{2}-16\,\kinvar_{2}\,\kinvar_{3}-
10\,\kinvar_{2}\,\kinvar_{4}-5\,{\kinvar_{3}}^{2}+20\,\kinvar_{3}\,\kinvar_{4}+{\kinvar_{4}}^{2}+2\,\kinvar_{1}+2\,\kinvar_{2}
 \right. \nonumber \\
 & & \left.
-4\,\kinvar_{3}
+2\,\kinvar_{4}+1 \right) {
\eps}^{2}\kinvar_{4}-3\, \left( \kinvar_{1}\,\kinvar_{2}-2\,\kinvar_{1}\,\kinvar_{3}+{\kinvar_{2}}^{2}-5\,\kinvar_{2}\,\kinvar_{3}-3\,\kinvar_{2}\,\kinvar_{4}-2\,{\kinvar_{3}}^
{2}+6\,\kinvar_{3}\,\kinvar_{4}+\kinvar_{2}-2\,\kinvar_{3} \right) \eps\,\kinvar_{4},
 \nonumber \\
c_{12,10} & = & 6\, \left( \kinvar_{1}\,\kinvar_{3}+\kinvar_{1}\,\kinvar_{4}-\kinvar_{2}\,\kinvar_{3}-\kinvar_{2}\,\kinvar_{4}+{\kinvar_{3}}^{2}+6\,\kinvar_{3}\,\kinvar_{4}+{\kinvar_{4}}^{
2}+\kinvar_{3}+\kinvar_{4} \right) \eps\,\kinvar_{1}\,\kinvar_{2},
 \nonumber \\
c_{12,11} & = & 6\, \left( \kinvar_{1}\,\kinvar_{2}-2\,\kinvar_{1}\,\kinvar_{4}+{\kinvar_{2}}^{2}-\kinvar_{2}\,\kinvar_{3}-3\,\kinvar_{2}\,\kinvar_{4}+2\,\kinvar_{3}\,\kinvar_{4}-2\,
{\kinvar_{4}}^{2}+\kinvar_{2}-2\,\kinvar_{4} \right) \eps\,\kinvar_{1}\,\kinvar_{3},
 \nonumber \\
c_{12,12} & = & 6\, \left( \kinvar_{1}\,\kinvar_{2}-2\,\kinvar_{1}\,\kinvar_{3}+{\kinvar_{2}}^{2}-3\,\kinvar_{2}\,\kinvar_{3}-\kinvar_{2}\,\kinvar_{4}-2\,{\kinvar_{3}}^{2}+2\,
\kinvar_{3}\,\kinvar_{4}+\kinvar_{2}-2\,\kinvar_{3} \right) \eps\,\kinvar_{1}\,\kinvar_{4},
 \nonumber \\
c_{12,13} & = & 3\, \left( {\kinvar_{1}}^{2}+2\,\kinvar_{1}\,\kinvar_{2}+2\,\kinvar_{1}\,\kinvar_{3}+{\kinvar_{2}}^{2}-2\,\kinvar_{2}\,\kinvar_{3}-4\,\kinvar_{2}\,\kinvar_{4}+{
\kinvar_{3}}^{2}+8\,\kinvar_{3}\,\kinvar_{4}-{\kinvar_{4}}^{2}+2\,\kinvar_{1}+2\,\kinvar_{2}+2\,\kinvar_{3}+1 \right) \eps\,\kinvar_{2}\,\kinvar_{3},
 \nonumber \\
c_{12,14} & = & 3\, \left( {\kinvar_{1}}^{2}+2\,\kinvar_{1}\,\kinvar_{2}+2\,\kinvar_{1}\,\kinvar_{4}+{\kinvar_{2}}^{2}-4\,\kinvar_{2}\,\kinvar_{3}-2\,\kinvar_{2}\,\kinvar_{4}-{
\kinvar_{3}}^{2}+8\,\kinvar_{3}\,\kinvar_{4}+{\kinvar_{4}}^{2}+2\,\kinvar_{1}+2\,\kinvar_{2}+2\,\kinvar_{4}+1 \right) \eps\,\kinvar_{2}\,\kinvar_{4},
 \nonumber \\
c_{12,15} & = & -6\, \left( {\kinvar_{1}}^{2}+2\,\kinvar_{1}\,\kinvar_{3}+2\,\kinvar_{1}\,\kinvar_{4}-{\kinvar_{2}}^{2}+2\,\kinvar_{2}\,\kinvar_{3}+2\,\kinvar_{2}\,\kinvar_{4}+
{\kinvar_{3}}^{2}-2\,\kinvar_{3}\,\kinvar_{4}+{\kinvar_{4}}^{2}+2\,\kinvar_{1}+2\,\kinvar_{3}+2\,\kinvar_{4}+1 \right) \eps\,\kinvar_{3}\,\kinvar_{4},
 \nonumber \\
c_{13,5} & = & -3\, \left( \kinvar_{3}-\kinvar_{4} \right)  \left( 17\,\kinvar_{1}-7\,\kinvar_{2}+\kinvar_{3}+\kinvar_{4}+1 \right) {\eps}^{3}- \left( 
\kinvar_{3}-\kinvar_{4} \right)  \left( 29\,\kinvar_{1}-13\,\kinvar_{2}+\kinvar_{3}+\kinvar_{4}+1 \right) {\eps}^{2}
 \nonumber \\
 & & 
-2\, \left( 2\,\kinvar_{1}-\kinvar_{2}
 \right)  \left( \kinvar_{3}-\kinvar_{4} \right) \eps,
 \nonumber \\
c_{13,6} & = & 6\, \left( \kinvar_{3}-\kinvar_{4} \right)  \left( 3\,\kinvar_{1}-3\,\kinvar_{2}-\kinvar_{3}-\kinvar_{4}-1 \right) {\eps}^{2}\kinvar_{1}+2\,
 \left( 3\,\kinvar_{1}-3\,\kinvar_{2}-\kinvar_{3}-\kinvar_{4}-1 \right)  \left( \kinvar_{3}-\kinvar_{4} \right) \eps\,\kinvar_{1},
 \nonumber \\
c_{13,7} & = & 6\, \left( \kinvar_{3}-\kinvar_{4} \right)  \left( 5\,\kinvar_{1}-\kinvar_{2}+\kinvar_{3}+\kinvar_{4}+1 \right) {\eps}^{2}\kinvar_{2}+ \left( 
\kinvar_{3}-\kinvar_{4} \right)  \left( 9\,\kinvar_{1}-3\,\kinvar_{2}+\kinvar_{3}+\kinvar_{4}+1 \right) \eps\,\kinvar_{2},
 \nonumber \\
c_{13,8} & = & - \left( 7\,{\kinvar_{1}}^{2}+2\,\kinvar_{1}\,\kinvar_{2}-18\,\kinvar_{1}\,\kinvar_{3}+32\,\kinvar_{1}\,\kinvar_{4}-5\,{\kinvar_{2}}^{2}+6\,\kinvar_{2}\,
\kinvar_{3}-16\,\kinvar_{2}\,\kinvar_{4}-{\kinvar_{3}}^{2}+{\kinvar_{4}}^{2}+6\,\kinvar_{1}-6\,\kinvar_{2}-2\,\kinvar_{3}
 \right. \nonumber \\
 & & \left.
-1 \right) {\eps}^{2}\kinvar_{3}- \left( 2
\,\kinvar_{1}-\kinvar_{2} \right)  \left( \kinvar_{1}+\kinvar_{2}-3\,\kinvar_{3}+5\,\kinvar_{4}+1 \right) \eps\,\kinvar_{3},
 \nonumber \\
c_{13,9} & = &  \left( 7\,{\kinvar_{1}}^{2}+2\,\kinvar_{1}\,\kinvar_{2}+32\,\kinvar_{1}\,\kinvar_{3}-18\,\kinvar_{1}\,\kinvar_{4}-5\,{\kinvar_{2}}^{2}-16\,\kinvar_{2}\,
\kinvar_{3}+6\,\kinvar_{2}\,\kinvar_{4}+{\kinvar_{3}}^{2}-{\kinvar_{4}}^{2}+6\,\kinvar_{1}-6\,\kinvar_{2}-2\,\kinvar_{4}-1 \right) {\eps}^{2}\kinvar_{4}
 \nonumber \\
 & & 
+ \left( 2\,
\kinvar_{1}-\kinvar_{2} \right)  \left( \kinvar_{1}+\kinvar_{2}+5\,\kinvar_{3}-3\,\kinvar_{4}+1 \right) \eps\,\kinvar_{4},
 \nonumber \\
c_{13,10} & = & -2\, \left( \kinvar_{3}-\kinvar_{4} \right)  \left( 3\,\kinvar_{1}-3\,\kinvar_{2}-\kinvar_{3}-\kinvar_{4}-1 \right) \eps\,\kinvar_{1}\,\kinvar_{2},
 \nonumber \\
c_{13,11} & = & 2\, \left( {\kinvar_{1}}^{2}+\kinvar_{1}\,\kinvar_{2}-2\,\kinvar_{1}\,\kinvar_{3}+4\,\kinvar_{1}\,\kinvar_{4}+3\,\kinvar_{2}\,\kinvar_{3}-3\,\kinvar_{2}\,\kinvar_{4}+
{\kinvar_{3}}^{2}-{\kinvar_{4}}^{2}+2\,\kinvar_{1}+\kinvar_{2}+2\,\kinvar_{3}+1 \right) \eps\,\kinvar_{1}\,\kinvar_{3},
 \nonumber \\
c_{13,12} & = & -2\, \left( {\kinvar_{1}}^{2}+\kinvar_{1}\,\kinvar_{2}+4\,\kinvar_{1}\,\kinvar_{3}-2\,\kinvar_{1}\,\kinvar_{4}-3\,\kinvar_{2}\,\kinvar_{3}+3\,\kinvar_{2}\,\kinvar_{4}
-{\kinvar_{3}}^{2}+{\kinvar_{4}}^{2}+2\,\kinvar_{1}+\kinvar_{2}+2\,\kinvar_{4}+1 \right) \eps\,\kinvar_{1}\,\kinvar_{4},
 \nonumber \\
c_{13,13} & = &  \left( 3\,{\kinvar_{1}}^{2}+2\,\kinvar_{1}\,\kinvar_{2}-6\,\kinvar_{1}\,\kinvar_{3}+12\,\kinvar_{1}\,\kinvar_{4}-{\kinvar_{2}}^{2}+2\,\kinvar_{2}\,\kinvar_{3}-
4\,\kinvar_{2}\,\kinvar_{4}-{\kinvar_{3}}^{2}+{\kinvar_{4}}^{2}+2\,\kinvar_{1}-2\,\kinvar_{2}-2\,\kinvar_{3}-1 \right) \eps\,\kinvar_{2}\,\kinvar_{3},
 \nonumber \\
c_{13,14} & = & - \left( 3\,{\kinvar_{1}}^{2}+2\,\kinvar_{1}\,\kinvar_{2}+12\,\kinvar_{1}\,\kinvar_{3}-6\,\kinvar_{1}\,\kinvar_{4}-{\kinvar_{2}}^{2}-4\,\kinvar_{2}\,\kinvar_{3}
+2\,\kinvar_{2}\,\kinvar_{4}+{\kinvar_{3}}^{2}-{\kinvar_{4}}^{2}+2\,\kinvar_{1}-2\,\kinvar_{2}-2\,\kinvar_{4}-1 \right) \eps\,\kinvar_{2}\,\kinvar_{4},
 \nonumber \\
c_{13,15} & = & -4\, \left( \kinvar_{3}-\kinvar_{4} \right)  \left( 2\,\kinvar_{1}-\kinvar_{2} \right) \eps\,\kinvar_{3}\,\kinvar_{4},
 \nonumber \\
c_{14,5} & = & -12\, \left( \kinvar_{3}-\kinvar_{4} \right)  \left( \kinvar_{1}+7\,\kinvar_{2}+\kinvar_{3}+\kinvar_{4}+1 \right) {\eps}^{3}-3\, \left( 
\kinvar_{3}-\kinvar_{4} \right)  \left( \kinvar_{1}+15\,\kinvar_{2}+\kinvar_{3}+\kinvar_{4}+1 \right) {\eps}^{2}-6\, \left( \kinvar_{3}-\kinvar_{4} \right) 
\eps\,\kinvar_{2},
 \nonumber \\
c_{14,6} & = & 6\, \left( \kinvar_{3}-\kinvar_{4} \right)  \left( \kinvar_{1}+7\,\kinvar_{2}+\kinvar_{3}+\kinvar_{4}+1 \right) {\eps}^{2}\kinvar_{1}+12\,
 \left( \kinvar_{3}-\kinvar_{4} \right) \eps\,\kinvar_{1}\,\kinvar_{2},
 \nonumber \\
c_{14,7} & = & -6\, \left( \kinvar_{3}-\kinvar_{4} \right)  \left( \kinvar_{1}-5\,\kinvar_{2}+\kinvar_{3}+\kinvar_{4}+1 \right) {\eps}^{2}\kinvar_{2}-3\,
 \left( \kinvar_{3}-\kinvar_{4} \right)  \left( \kinvar_{1}-3\,\kinvar_{2}+\kinvar_{3}+\kinvar_{4}+1 \right) \eps\,\kinvar_{2},
 \nonumber \\
c_{14,8} & = & 3\, \left( {\kinvar_{1}}^{2}-2\,\kinvar_{1}\,\kinvar_{2}+2\,\kinvar_{1}\,\kinvar_{3}-3\,{\kinvar_{2}}^{2}+10\,\kinvar_{2}\,\kinvar_{3}-16\,\kinvar_{2}\,
\kinvar_{4}+{\kinvar_{3}}^{2}-{\kinvar_{4}}^{2}+2\,\kinvar_{1}-2\,\kinvar_{2}+2\,\kinvar_{3}+1 \right) {\eps}^{2}\kinvar_{3}
 \nonumber \\
 & & 
-3\, \left( \kinvar_{1}+\kinvar_{2}-3\,
\kinvar_{3}+5\,\kinvar_{4}+1 \right) \eps\,\kinvar_{2}\,\kinvar_{3},
 \nonumber \\
c_{14,9} & = & -3\, \left( {\kinvar_{1}}^{2}-2\,\kinvar_{1}\,\kinvar_{2}+2\,\kinvar_{1}\,\kinvar_{4}-3\,{\kinvar_{2}}^{2}-16\,\kinvar_{2}\,\kinvar_{3}+10\,\kinvar_{2}\,
\kinvar_{4}-{\kinvar_{3}}^{2}+{\kinvar_{4}}^{2}+2\,\kinvar_{1}-2\,\kinvar_{2}+2\,\kinvar_{4}+1 \right) {\eps}^{2}\kinvar_{4}
 \nonumber \\
 & & 
+3\, \left( \kinvar_{1}+\kinvar_{2}+5\,
\kinvar_{3}-3\,\kinvar_{4}+1 \right) \eps\,\kinvar_{2}\,\kinvar_{4},
 \nonumber \\
c_{14,10} & = & 6\, \left( \kinvar_{3}-\kinvar_{4} \right)  \left( \kinvar_{1}-\kinvar_{2}+\kinvar_{3}+\kinvar_{4}+1 \right) \eps\,\kinvar_{1}\,\kinvar_{2},
 \nonumber \\
c_{14,11} & = & 6\, \left( \kinvar_{1}+\kinvar_{2}-\kinvar_{3}+3\,\kinvar_{4}+1 \right) \eps\,\kinvar_{1}\,\kinvar_{2}\,\kinvar_{3},
 \nonumber \\
c_{14,12} & = & -6\, \left( \kinvar_{1}+\kinvar_{2}+3\,\kinvar_{3}-\kinvar_{4}+1 \right) \eps\,\kinvar_{1}\,\kinvar_{2}\,\kinvar_{4},
 \nonumber \\
c_{14,13} & = & 3\, \left( {\kinvar_{1}}^{2}+2\,\kinvar_{1}\,\kinvar_{2}+2\,\kinvar_{1}\,\kinvar_{3}+{\kinvar_{2}}^{2}-2\,\kinvar_{2}\,\kinvar_{3}+4\,\kinvar_{2}\,\kinvar_{4}+{
\kinvar_{3}}^{2}-{\kinvar_{4}}^{2}+2\,\kinvar_{1}+2\,\kinvar_{2}+2\,\kinvar_{3}+1 \right) \eps\,\kinvar_{2}\,\kinvar_{3},
 \nonumber \\
c_{14,14} & = & -3\, \left( {\kinvar_{1}}^{2}+2\,\kinvar_{1}\,\kinvar_{2}+2\,\kinvar_{1}\,\kinvar_{4}+{\kinvar_{2}}^{2}+4\,\kinvar_{2}\,\kinvar_{3}-2\,\kinvar_{2}\,\kinvar_{4}-
{\kinvar_{3}}^{2}+{\kinvar_{4}}^{2}+2\,\kinvar_{1}+2\,\kinvar_{2}+2\,\kinvar_{4}+1 \right) \eps\,\kinvar_{2}\,\kinvar_{4},
 \nonumber \\
c_{14,15} & = & -12\, \left( \kinvar_{3}-\kinvar_{4} \right) \eps\,\kinvar_{2}\,\kinvar_{3}\,\kinvar_{4}.
\eq
}
These expressions are also available through a supplementary electronic file
attached to the arxiv version of this article.

\section{Determining the functions of the ansatz}
\label{appendix:dgl_step_2}

In this appendix we give a detailed discussion for the transformation from an
$F^\bullet$-compatible differential equation to an $\eps$-factorised differential equation in the case 
where we have three non-trivial parts in the filtration
\bq
 \emptyset = F^3 V \subseteq F^2 V \subseteq F^1 V \subseteq F^0 V = V.
\eq
We set
\bq
 d_1 \; = \; \dim \mathrm{Gr}_F^2 V,
 \;\;\;\;\;\;
 d_2 \; = \; \dim \mathrm{Gr}_F^1 V,
 \;\;\;\;\;\;
 d_3 \; = \; \dim \mathrm{Gr}_F^0 V.
\eq
The matrix $A$ appearing in the differential equation $dJ = \hat{A}J$ is then of the form
\bq
 \hat{A}
 & = &
 \left( \begin{array}{rrr}
 A^{(0)}_{11} + \eps A^{(1)}_{11} & \eps A^{(1)}_{12} & 0 \\
 \frac{1}{\eps} A^{(-1)}_{21} + A^{(0)}_{21} + \eps A^{(1)}_{21} & A^{(0)}_{22} + \eps A^{(1)}_{22} & \eps A^{(1)}_{23} \\
 \frac{1}{\eps^2} A^{(-2)}_{31} + \frac{1}{\eps} A^{(-1)}_{31} + A^{(0)}_{31} + \eps A^{(1)}_{31} & \frac{1}{\eps} A^{(-1)}_{32} + A^{(0)}_{32} + \eps A^{(1)}_{32} & A^{(0)}_{33} + \eps A^{(1)}_{33} \\
 \end{array} \right),
\eq
where the entries are in general matrices.
We have $\dim A^{(j)}_{11} = d_1$, $\dim A^{(j)}_{22} = d_2$ and $\dim A^{(j)}_{33} = d_3$.
The off-diagonal entries have the dimensions dictated by the diagonal blocks.
In the case of the unequal-mass three-loop banana integral, 
$A^{(j)}_{11}$ is a $(1 \times 1)$-matrix,
$A^{(j)}_{22}$ is a $(8 \times 8)$-matrix and
$A^{(j)}_{33}$ is a $(2 \times 2)$-matrix.
It is worth discussing the case with three non-trivial parts in the filtration in full generality.
We construct a transformation
\bq
 R_2 & = & R_2^{(-2)} R_2^{(-1)} R_2^{(0)},
\eq
leading to the $\eps$-factorised basis $K=R_2^{-1} J$.
The ansatz for $R_2^{(-2)}$reads
\bq
 R_2^{(-2)}
 & = &
 \left( \begin{array}{rrr}
 R^{(0)}_{11} & 0 & 0 \\
 \frac{1}{\eps} R^{(-1)}_{21} & R^{(0)}_{22} & 0 \\
 \frac{1}{\eps^2} R^{(-2)}_{31} & \frac{1}{\eps} R^{(-1)}_{32} & R^{(0)}_{33} \\
 \end{array} \right).
\eq
For a basis transformation the determinant of $R_2^{(-2)}$ has to be non-zero, this implies that 
$R^{(0)}_{11}$, $R^{(0)}_{22}$ and $R^{(0)}_{33}$ are invertible matrices.

Requiring that terms of $B$-order $(-2)$ vanish, gives us six equations.
These equations group into $6=3+2+1$ as follows:
The first group of three equations involves only $R^{(0)}_{11}$, $R^{(0)}_{21}$ and $R^{(0)}_{31}$:
\bq
 d R^{(0)}_{11} & = & A^{(0)}_{11} R^{(0)}_{11} + A^{(1)}_{12} R^{(-1)}_{21},
 \nonumber \\
 d R^{(-1)}_{21} & = & A^{(-1)}_{21} R^{(0)}_{11} + A^{(0)}_{22} R^{(-1)}_{21} + A^{(1)}_{23} R^{(-2)}_{31},
 \nonumber \\
 d R^{(-2)}_{31} & = & A^{(-2)}_{31} R^{(0)}_{11} + A^{(-1)}_{32} R^{(-1)}_{21} + A^{(0)}_{33} R^{(-2)}_{31}.
\eq
The second set involves in addition $R^{(0)}_{22}$ and $R^{(0)}_{32}$:
\bq
 d R^{(0)}_{22} & = & A^{(0)}_{22} R^{(0)}_{22} + A^{(1)}_{23} R^{(-1)}_{32} - R^{(-1)}_{21} \left( R^{(0)}_{11} \right)^{-1} A^{(1)}_{12} R^{(0)}_{22},
 \nonumber \\
 d R^{(-1)}_{32} & = & A^{(-1)}_{32} R^{(0)}_{22} + A^{(0)}_{33} R^{(-1)}_{32} - R^{(-2)}_{31} \left( R^{(0)}_{11} \right)^{-1} A^{(1)}_{12} R^{(0)}_{22}.
\eq
The last equation involves in addition $R^{(0)}_{33}$:
\bq
 d R^{(0)}_{33} & = & A^{(0)}_{33} R^{(0)}_{33} - R^{(-1)}_{32} \left( R^{(0)}_{22} \right)^{-1} A^{(1)}_{23} R^{(0)}_{33}.
\eq
We then turn to $R_2^{(-1)}$. The ansatz for $R_2^{(-1)}$ reads
\bq
 R_2^{(-1)}
 & = &
 \left( \begin{array}{rrr}
 1 & 0 & 0 \\
 R^{(0)}_{21} & 1 & 0 \\
 \frac{1}{\eps} R^{(-1)}_{31} & R^{(0)}_{32} & 1 \\
 \end{array} \right).
\eq
Requiring that terms of $B$-order $(-1)$ vanish, gives us three equations, which split into $3=1+1+1$.
At each stage, the right-hand side involves only functions already known at this stage.
The equations are
\bq
 d R^{(-1)}_{31} & = & \left( R^{(0)}_{33} \right)^{-1} \left\{ A^{(-1)}_{31} R^{(0)}_{11} + A^{(0)}_{32} R^{(-1)}_{21} + A^{(1)}_{33} R^{(-2)}_{31}  
                   - R^{(-1)}_{32} \left( R^{(0)}_{22} \right)^{-1} \left( A^{(0)}_{21} R^{(0)}_{11} 
 \right. \right. \nonumber \\
 & & \left. \left.
                   + A^{(1)}_{22} R^{(-1)}_{21} \right)
                   - \left[ R^{(-2)}_{31} - R^{(-1)}_{32} \left( R^{(0)}_{22} \right)^{-1} R^{(-1)}_{21} \right] \left( R^{(0)}_{11} \right)^{-1} A^{(1)}_{11} R^{(0)}_{11} \right\},
 \nonumber \\
 d R^{(0)}_{21} & = & \left( R^{(0)}_{22} \right)^{-1} \left( A^{(0)}_{21} R^{(0)}_{11} + A^{(1)}_{22} R^{(-1)}_{21} + A^{(1)}_{23} R^{(0)}_{33} R^{(-1)}_{31}  
                   - R^{(-1)}_{21} \left( R^{(0)}_{11} \right)^{-1} A^{(1)}_{11} R^{(0)}_{11} \right),
 \nonumber \\
 d R^{(0)}_{32} & = & \left( R^{(0)}_{33} \right)^{-1} \left( A^{(0)}_{32} R^{(0)}_{22} + A^{(1)}_{33} R^{(-1)}_{32}
                   - R^{(-1)}_{32} \left( R^{(0)}_{22} \right)^{-1} A^{(1)}_{22} R^{(0)}_{22} \right)
 \nonumber \\
 & &
                   - R^{(-1)}_{31} \left( R^{(0)}_{11} \right)^{-1} A^{(1)}_{12} R^{(0)}_{22}.
\eq
Finally, we consider the matrix $R_2^{(0)}$. The ansatz for $R_2^{(0)}$ reads
\bq
 R^{(0)}
 & = &
 \left( \begin{array}{rrr}
 1 & 0 & 0 \\
 0 & 1 & 0 \\
 R^{(0)}_{31} & 0 & 1 \\
 \end{array} \right).
\eq
Requiring that terms of $B$-order $0$ vanish, gives us one equation:
\bq
 d R^{(0)}_{31} & = & 
 \left( R^{(0)}_{33} \right)^{-1} 
   \left[ A^{(0)}_{31} R^{(0)}_{11} 
          + A^{(1)}_{32} R^{(-1)}_{21} 
          + A^{(0)}_{32} R^{(0)}_{22} R^{(0)}_{21} 
          + A^{(1)}_{33} R^{(-1)}_{32} R^{(0)}_{21} 
          + A^{(1)}_{33} R^{(0)}_{33} R^{(-1)}_{31} 
 \right. \nonumber \\
 & & \left.
          - R^{(-1)}_{32} \left( R^{(0)}_{22} \right)^{-1} \left( A^{(1)}_{21} R^{(0)}_{11} + A^{(1)}_{22} R^{(0)}_{22} R^{(0)}_{21} \right) 
 \right]
 \nonumber \\
 & &
 - R^{(-1)}_{31} \left( R^{(0)}_{11} \right)^{-1} \left( A^{(1)}_{11} R^{(0)}_{11} + A^{(1)}_{12} R^{(0)}_{22} R^{(0)}_{21} \right).
\eq

\end{appendix}

{\footnotesize
\bibliography{biblio}
\bibliographystyle{h-physrev5}
}

\end{document}